\pdfoutput=1
\documentclass[11pt]{article}
\usepackage[T1]{fontenc}
\usepackage{geometry} 
\usepackage[english]{babel}
\usepackage{cite}
\usepackage{amsmath}
\usepackage{amsthm}
\usepackage{amsfonts}
\usepackage{amssymb}
\usepackage{xspace} % \xspace
\usepackage{mathtools} % \DeclarePairedDelimiter
\usepackage{braket} % \set \Braket
\usepackage[linesnumbered,vlined,ruled]{algorithm2e} 
\usepackage{tikz} % \tikzpicture
\usepackage{thm-restate} % \begin{restatabke}
\usepackage{multirow} % \multirow
\usepackage{hyperref} %\hyperref
\usepackage[capitalise]{cleveref} % \cref

\graphicspath{ {./images/} }%helpful if your graphic files are in another directory

\usetikzlibrary{positioning} % of 

\newtheorem{theorem}{Theorem}[section]
\newtheorem{claim}[theorem]{Claim}
\newtheorem{lemma}[theorem]{Lemma}

\newtheorem{definition}[theorem]{Definition}

\newenvironment{claimproof}[1]{\begin{proof}}{\end{proof}}

\crefalias{AlgoLine}{line}
\makeatletter
\let\cref@old@stepcounter\stepcounter
\def\stepcounter#1{%
  \cref@old@stepcounter{#1}%
  \cref@constructprefix{#1}{\cref@result}%
  \@ifundefined{cref@#1@alias}%
    {\def\@tempa{#1}}%
    {\def\@tempa{\csname cref@#1@alias\endcsname}}%
  \protected@edef\cref@currentlabel{%
    [\@tempa][\arabic{#1}][\cref@result]%
    \csname p@#1\endcsname\csname the#1\endcsname}}
\makeatother

\newcommand{\Z}{\mathbb{Z}}
\newcommand{\N}{\mathbb{N}}

\DeclarePairedDelimiter{\floor}{\lfloor}{\rfloor}
\DeclarePairedDelimiter{\ceil}{\lceil}{\rceil}

\newcommand{\hybrid}{\ensuremath{\mathsf{Hybrid}}\xspace}
\newcommand{\ncc}{\ensuremath{\mathsf{Node\text{-}Capacitated\ Clique }}\xspace}
\newcommand{\congest}{\ensuremath{\mathsf{CONGEST}}\xspace}
\newcommand{\local}{\ensuremath{\mathsf{LOCAL}}\xspace}
\newcommand{\clique}{\ensuremath{\mathsf{Congested\ Clique}}\xspace}
\newcommand{\bcc}{\ensuremath{\mathsf{Broadcast\ Congested\ Clique}}\xspace}

\newcommand{\oracle}{\ensuremath{\mathsf{Oracle}}\xspace}
\newcommand{\tiered}{\ensuremath{\mathsf{Tiered\ Oracles}}\xspace}

\newcommand{\bigO}[1]{\ensuremath{{O}\left(#1\right)}}
\newcommand{\bigOmega}[1]{\ensuremath{{\Omega}\left(#1\right)}}

\newcommand{\tildeBigO}[1]{\ensuremath{{\tilde{O}}(#1)}}
\newcommand{\tildeBigOmega}[1]{\ensuremath{{\tilde{\Omega}}(#1)}}
\newcommand{\tildeTheta}[1]{\ensuremath{{\tilde{\Theta}}(#1)}}

\newcommand{\interval}[1]{\ensuremath{\left[#1\right]}}
\newcommand{\size}[1]{\ensuremath{\left|#1\right|}}

\makeatletter
\newcommand*{\whp}{%
    \@ifnextchar{.}%
        {w.h.p}%
        {w.h.p.\@\xspace}%
}
\makeatother

\newcommand*{\fullrefsingle}[1]{\hyperref[{#1}]{\cref*{#1} (\emph{\nameref*{#1}})}} % Claim # (Name)
\makeatletter
\newcommand{\fullref}[1]{%
  \@tempswafalse
  \@for\next:=#1\do
    {\if@tempswa\ and \else\@tempswatrue\fi\fullrefsingle{\next}}%
}
\makeatother

\newenvironment{proofof}[1]{\begin{proof}[Proof of #1]}{\end{proof}}

\SetKwFunction{TokenDissemination}{Token-Dissemination}
\SetKwFunction{ReassignSkeletons}{Reassign-Skeletons}

\tikzset{main node/.style={circle,fill=black!20,draw,minimum size=1cm,inner sep=0pt}}
\tikzset{block node/.style={rectangle,fill=white!20,minimum height=1.0cm,text width=2.8cm,align=center,inner sep=0pt}}
\tikzset{diedge/.style={->}}

\title{Distance Computations in the Hybrid Network Model via Oracle Simulations}

\author{Keren Censor-Hillel
    \and Dean Leitersdorf
    \and Volodymyr Polosukhin
}

\date{Technion\footnote{\{ckeren, leitersdorf, po\}@cs.technion.ac.il}}

\begin{document}

\begin{titlepage}
\maketitle
\thispagestyle{empty}

\begin{abstract}	
% Latex Does not work if tabs in front
The \hybrid network model was introduced in [Augustine et al., SODA '20] for laying down a theoretical foundation for networks which combine two possible modes of communication: One mode allows high-bandwidth communication with neighboring nodes, and the other allows low-bandwidth communication over few long-range connections at a time. This fundamentally abstracts networks such as hybrid data centers, and class-based software-defined networks.

Our technical contribution is a \emph{density-aware} approach that allows us to simulate a set of \emph{oracles} for an overlay skeleton graph over a \hybrid network. 

As applications of our oracle simulations, with additional machinery that we provide, we derive fast algorithms for fundamental distance-related tasks.  One of our core contributions is an algorithm in the \hybrid model for computing \emph{exact} weighted 
shortest paths from $\tilde O(n^{1/3})$ sources which completes in $\tilde O(n^{1/3})$ rounds w.h.p. 
This improves, in both the runtime and the number of sources, upon the algorithm of [Kuhn and Schneider, PODC ’20], which computes shortest paths from a single source in $\tilde O(n^{2/5})$ rounds w.h.p. 

We additionally show a 2-approximation for weighted diameter and a $(1+\epsilon)$-approximation for unweighted diameter, both in $\tilde O(n^{1/3})$ rounds w.h.p., which is comparable to the $\tilde \Omega(n^{1/3})$ lower bound of [Kuhn and Schneider, PODC ’20] for a $(2-\epsilon)$-approximation for  weighted diameter and an exact unweighted diameter. 
We also provide fast 
distance \emph{approximations} from multiple sources and fast approximations for eccentricities. 
\end{abstract}
\end{titlepage}

\tableofcontents
\thispagestyle{empty} % No number on TOC Page
\newpage

\setcounter{page}{1}
\section{Introduction}\label{sec:introduction}

The \hybrid model of computation was recently introduced by Augustine et al.~\cite{AHKSS20}, for abstracting networks which can utilize both high-bandwidth local communication links, as well as very few low-bandwidth global communication channels. This model abstracts fundamental systems, such as  
a combination of device-to-device communication with cellular networks (e.g. 5G) \cite{Asadi2016AnSDR},
wired data centers with wireless links  (hybrid DCNs) \cite{Wang2010CThrough,Cui2011ChannelAI,Huang2013TheAA,Han2015RUSHRA},
and Class-Based Hybrid software-defined networks (SDNs) \cite{Vissicchio2017Opportunities}.

The pioneering works of \cite{AHKSS20, kuhn2020computing, feldmann2020fast} provide fast algorithms for various distance-related tasks in the \hybrid model. At the heart of many of these algorithms lies a framework for using skeleton overlay graphs for computation and approximation of distances, as well as for fast communication. 

In this paper, we  define and show how to efficiently simulate \emph{oracles} over skeleton graphs in the \hybrid model. Using additional machinery that we provide, the implications of our simulations are faster algorithms for distance computations in the \hybrid model. Our oracle models could also be of independent interest, presenting a generic approach which can potentially be applied elsewhere.

\subsection{Our Contributions}
The \hybrid model, which we consider in this paper, abstracts a synchronous network of nodes, where in each round, every node can send and receive 
arbitrarily many messages of $O(\log n)$ bits to/from each of its neighbors (over \emph{local edges}) and an additional $O(\log n)$ messages, in total, to/from any other nodes in the network (over \emph{global edges}). The high bandwidth permissible over the local edges is aligned with previous research in the \hybrid model as well as with the extensively studied \local distributed model.

The main idea which our
results
hinge upon is exploiting an inherent \emph{asymmetry} in the \hybrid model which we observe. This asymmetry allows nodes with \emph{dense} neighborhoods to effectively receive significantly more information. 
To see this, note that every node can use the global edges of the \hybrid model to send and receive a limited number of messages every round. However, since in the next round a node can communicate with its neighborhood in the graph using local edges and share with them the information which it received, 
this implies that a node is able to learn much more information from the entire graph if it is in a more dense neighborhood. Thus, \emph{density-aware} algorithms are inherently useful in the \hybrid model.

We use this asymmetry to simulate \oracle and \tiered models. To capture what an oracle can do, we introduce the \oracle and \tiered models. Roughly speaking, in the \oracle model there is a node $\ell$, the \emph{oracle}, which can receive $\deg(v)$ messages from each node $v$, within \emph{a single round}. In particular, this implies that the oracle can learn the entire communication graph.

We cannot afford to directly simulate the \oracle model as it requires too much communication in the \hybrid model. Instead, we simulate the \oracle model over a \emph{skeleton graph}. Roughly speaking, given an input graph $G$, a skeleton graph is a subset of the nodes of $G$, connected by virtual edges which represent paths in $G$. Skeleton graphs are a common tool for distance computations in various models \cite{lenzen2019distributedDC,AHKSS20,kuhn2020computing,Ullman1990HighprobabilityPT,Roditty2011OnDynamicSPP}, and it has been shown in \cite{AHKSS20,kuhn2020computing} that some distances on the skeleton graph can be efficiently extended to distances on the entire graph in the \hybrid model. 

As a warm-up, we show that a single round of the \oracle model over certain skeleton graphs can be simulated in $\tilde O(n^{1/3})$ rounds, \whp\footnote{As common, \whp indicates a probability that is at least $1-n^{-c}$, for some constant $c>1$.}, in the \hybrid model. Combining this with a simple, constant-round algorithm for exact weighted \emph{single source shortest paths (SSSP)} which we show in the \oracle model, gives the following theorem.

\begin{restatable}[Exact SSSP]{theorem}{exactSSSP}
\label{theorem:exactSSSP}
Given a weighted graph $G=(V, E)$, there is an algorithm in the \hybrid model that computes an exact weighted SSSP in $\tildeBigO{n^{1/3}}$ rounds \whp
\end{restatable}

This result should be compared with the previous state-of-the-art algorithms for exact weighted SSSP in $\tilde O(n^{2/5})$ rounds~\cite{kuhn2020computing}, and in $\tilde O(\sqrt{SPD})$ rounds~\cite{AHKSS20}, where $SPD$ is the length of the shortest path diameter. Further, it improves upon the $\tilde O(n^{1/3}/\epsilon^6)$ round algorithm for a ($1+\epsilon$)-approximation of weighted SSSP~\cite{AHKSS20}, in both the runtime and in being exact. We stress that this is a warm up, and later on we extend this result to shortest path distances from $O(n^{1/3})$ sources, instead of a single source, in the same round complexity of $\tilde O(n^{1/3})$.

It is well known that one can approximate the diameter using a solution to SSSP, and so as a byproduct we get the following result. 

\begin{restatable}[$2$-Approx. Weighted Diameter]{corollary}{weightedDiameter}\label{thr:weightedDiameter}
    There is an algorithm in the \hybrid model that computes a $2$-approximation of weighted diameter in $\tildeBigO{n^{1/3}}$ rounds \whp
\end{restatable}

Notably, $\tilde \Omega(n^{1/3})$ rounds are necessary for a $(2-\epsilon)$-approximation for weighted diameter~\cite{kuhn2020computing}. Our algorithm in \Cref{thr:weightedDiameter} thus raises the interesting open question of whether one can go below this complexity for a 2-approximation.

While efficiently simulating an oracle is powerful, it still does not exploit the full capacity of the \hybrid model. This observation brings us to enhance the \oracle model and introduce the \tiered model, which consists of multiple oracles with varying abilities. In a nutshell, in the \tiered model, in each round every node $v$ can send (the same) $\deg(v)$ messages to all nodes $u$ with $\deg(u)\geq\deg(v)/2$. This basically means that nodes are  bucketed according to degrees and each node is an oracle for all nodes in buckets below it. One can notice that the node with the highest degree in the graph is equivalent to the oracle in the \oracle model, but here, the other nodes in the graph also have some \emph{partial} oracle capabilities.

We show how to simulate the \tiered model over skeleton graphs in the \hybrid model within $\tilde O(n^{1/3})$ rounds. Subsequently, we present an algorithm which solves all pairs shortest paths (APSP) using one round of the \tiered model and $O(\log n)$ rounds of the \clique model\footnote{The \clique is a synchronous distributed model where every two nodes in the graph can exchange messages of $O(\log n)$ bits in every round.}. We then utilize our \tiered model simulation, along with a previously known simulation of the \clique model from \cite{kuhn2020computing}, to simulate the APSP algorithm over skeleton graphs in the \hybrid model. Our efficient computation of APSP over a skeleton graph in the \hybrid model then leads to computing multi-source shortest paths from \emph{random} sources in the \hybrid model. 

Shortest paths from random sources is a crucial stepping stone for our later results. We show that computing shortest path distances from random sources to the entire graph, allows us to subsequently obtain fast algorithms for other distance problems. 
We call the problem of computing distances from sources sampled with probability $n^{x-1}$ i.i.d $n^{x}$-RSSP.

\begin{restatable}[$n^{x}$-RSSP]{theorem}{exactRMSSP}\label{thr:exactRMSSP}
Given a graph $G=(V, E)$, $0 < x < 1$, and a set of nodes $M$ sampled independently with probability $n^{x-1}$, there is an algorithm in the \hybrid model that ensures that every $v \in V$ knows the exact, weighted distance from itself to every node in $M$ within \tildeBigO{n^{1/3}+n^{2x-1}} rounds \whp
\end{restatable}

We complement \Cref{thr:exactRMSSP} with a lower bound, following the lines of \cite{AHKSS20,kuhn2020computing}, for approximating distances from many random sources, to any reasonable approximation factor, which tightly matches the upper bound when $x=2/3$. 

\begin{restatable}[Lower Bound Exact Shortest Paths, Sources Sampled i.i.d.]{theorem}{lowerBound}\label{thr:lowerBound}
    Let $p=\bigOmega{\log n/n}$ and $\alpha< \sqrt{n/p}\cdot\log(n)$/2. Any $\alpha$-approximate unweighted algorithm from random sources sampled independently with probability $p$ in the \hybrid network model takes $\bigOmega{\sqrt{p\cdot n}/\log n}$ rounds \whp 
\end{restatable}

We leverage our \emph{near-optimal}  (tight up to polylogarithmic factors) algorithm for shortest paths from a set $M$ of $\tildeBigO{n^{2/3}}$ random sources in order to obtain exact weighted shortest paths from any \emph{given} set $U$ of $O(n^{1/3})$ sources. We achieve this by adapting the behavior of the given fixed source nodes to the density of their neighborhoods, as follows. A source node $s \in U$ in a sparse neighborhood broadcasts the distances to all the random source nodes from $M$ it sees in its neighborhood. A source node $s \in U$ in a dense neighborhood takes control of one of the random sources in $M$ in its neighborhood and uses it as a proxy in order to communicate enough information to all the other nodes in the graph so that they could determine their distances from $s$. We remark that this proxy approach is a \emph{key insight} which we later encapsulate as a general tool in the \hybrid model and may potentially be of independent interest.
Our approach gives the following theorem.

\begin{restatable}[Exact $n^{1/3}$ Sources Shortest Paths]{theorem}{exactMSSP}\label{thr:exactMSSP}
    Given a weighted graph $G=(V, E)$, and a set of sources $U$, such that $|U| = O(n^{1/3})$, there exists an algorithm, at the end of which each $v \in V$ knows its distance from every $s\in U$, which runs in \tildeBigO{n^{1/3}} rounds \whp
\end{restatable}

\Cref{thr:exactMSSP} raises an interesting open question of whether the complexity of SSSP in the \hybrid model is below that of computing shortest paths from \tildeBigO{n^{1/3}} sources.

We also exploit our aforementioned solution for computing APSP on the skeleton graph to obtain approximate distances from a larger set of given sources ($n^x$-SSP), as follows.

\begin{restatable}[Approximate Multiple Source Shortest Paths]{theorem}{MSSP}\label{thr:MSSP}
Given a graph $G=(V, E)$, a set of sources $U$, where $|U| = \tilde \Theta(n^y)$ for some constant $0 < y < 1$, and a value $0 < \epsilon$, there is an algorithm in the \hybrid model which ensures that every node $v \in V$ knows an approximation to its distance from every $s \in U$, where the approximation factor is $(1+\epsilon)$ if $G$ is unweighted and $3$ if $G$ is weighted. The  complexity of the algorithm is $\tildeBigO{n^{1/3}/\epsilon+n^{y/2}}$ rounds, \whp
\end{restatable}

This result improves both in round complexity and approximation factors upon the previous results in \cite{kuhn2020computing}. The reason for this is that we compute APSP over skeleton graphs using the efficient, exact algorithm from the \tiered oracle model, while \cite{kuhn2020computing} simulate the slower, approximate algorithms of \cite{CensorHillel2016AlgebraicMI,CensorHillel2020FastAS} in the \clique model. Particularly, this result is tight up to polylogarithmic factors for $y\geq 2/3$ due to a lower bound of \cite{kuhn2020computing}.

We can also approximate unweighted eccentricities by a combination of computing shortest path distances from $n^{2/3}$ random sources and performing local explorations using the local edges of the model. For approximating weighted eccentricities, this is insufficient, and here our approach is to additionally broadcast required information from each random source node regarding its $\tildeBigO{n^{1/3}}$-hop neighborhood in the graph. We obtain the following result.

\begin{theorem}[Approx. Eccentricities]\label{thr:eccentricitiesApproximation}
Given a graph $G=(V, E)$, there is an algorithm in the \hybrid model that computes a $(1+\epsilon)$-approximation of unweighted and $3$-approximation of weighted eccentricities in $\tilde{O}(n^{1/3}/\epsilon)$ rounds, \whp 
\end{theorem}

Finally, the unweighted eccentricities approximation directly implies a $(1+\epsilon)$ approximation for unweighted diameter. This should be compared with the lower bound of $\tilde \Omega(n^{1/3})$ rounds for exact unweighted diameter due to \cite{kuhn2020computing}.

\begin{restatable}
[$(1+\epsilon)$-Approx. Unweighted Diameter]{corollary}{unweigtedDiameterApproximation}\label{thr:unweigtedDiameterApproximation}
Let $G=(V, E)$ be an unweighted graph, and let $\epsilon>0$. There exists an algorithm in the \hybrid model which computes a $(1+\epsilon)$-approximation of the diameter in $\tildeBigO{n^{1/3}/\epsilon}$ rounds, \whp
\end{restatable}

We refer the reader to \Cref{table:results}, for a visual summary of our end results with comparison to related work.

\textbf{Roadmap:} The remainder of the current section is dedicated towards surveying related work. In \cref{sec:prelim}, we provide all formal definitions. Next, in \cref{sec:superCubeRootAlgorithms:tool} we formally define the \oracle and \tiered models, and show how to simulate them over skeleton graphs in the \hybrid model. Finally, \cref{sec:superCubeRootAlgorithms}, gives our algorithms for distance problems in the oracle models and, using our simulations, also in the \hybrid model. 
Some additional results are deferred to the appendix. We show the approximation for shortest path from $n^{x}$ given sources in \cref{app:MSSP}. In \cref{app:ecc,app:diameter} we provide eccentricities and diameter approximations, respectively. We wrap up with our lower bound for shortest paths from sources sampled i.i.d. in \cref{sec:lowerBound}.

\begin{table}[ht!]
    \centering
    \begin{tabular}{ | c | c | c | c | c | c | }
        \hline
        Problem & Variant & Approximation & This work & Previous works \\ 
        \hline
        \multirow{3}{*}{SSSP}
            & weighted 
            & exact 
            & $\tildeBigO{n^{1/3}}$
            & $\tildeBigO{n^{2/5}}$\cite{kuhn2020computing}, $\tildeBigO{\sqrt{SPD}}$\cite{AHKSS20} 
            \\
        % SSSP 
            & weighted 
            & $1+\epsilon$ 
            & 
            & $\tildeBigO{n^{1/3}\cdot \epsilon^{-6}}$\cite{AHKSS20}
            \\
        % SSSP 
            & weighted 
            & $(1/\epsilon)^{O(1/\epsilon)}$ 
            & 
            & $n^{\epsilon}$\cite{AHKSS20}
            \\
            
        \hline
        \multirow{3}{*}{$n^{x}$-RSSP}
            & unweighted
            & $\tildeBigO{n^{1-x/2}}$
            & $\tildeBigOmega{n^{x/2}}$
            &
            \\
        % $n^{x}$-RSSP
            & weighted 
            & exact
            & $\tildeBigO{n^{1/3}+n^{2x-1}}$
            & 
            \\
        % $n^{x}$-RSSP
            & weighted 
            & $2+\epsilon$
            & 
            & $\tildeBigO{n^{1/3}+n^{2x-1}}$ \cite{kuhn2020computing}
            \\
        
        \hline
        \multirow{3}{*}{$n^{1/3}$-SSP}
            & unweighted 
            & $1 + \epsilon$
            & 
            & $\tildeBigO{n^{1/3}/\epsilon}$ \cite{kuhn2020computing}
            \\
        % $n^{1/3}$-SSP 
            & weighted 
            & exact
            & $\tildeBigO{n^{1/3}}$
            &
            \\
        % $n^{1/3}$-SSP 
            & weighted 
            & $3 + \epsilon$
            & 
            & $\tildeBigO{n^{1/3}/\epsilon}$ \cite{kuhn2020computing}
            \\
        \hline
        \multirow{6}{*}{$n^{x}$-SSP} 
            & unweighted
            & \tildeBigO{n^{1-x/2}}
            & 
            & $\tildeBigOmega{n^{x/2}}$\cite{kuhn2020computing} 
            \\
        % $n^{x}$-SSP 
            & unweighted
            & $1+\epsilon$ 
            & $\tildeBigO{n^{1/3}/\epsilon + n^{x/2}}$
            & 
            \\
        % $n^{x}$-SSP 
            & unweighted
            & $2+\epsilon$ 
            & 
            & $\tildeBigO{n^{1/3}/\epsilon+n^{x/2}}$ \cite{kuhn2020computing}
            \\  
            
        % \hline
        % $n^{x}$-SSP 
            & weighted
            & $3$ 
            & $\tildeBigO{n^{1/3}+n^{x/2}}$
            & 
            \\  
        % $n^{x}$-SSP 
            & weighted
            & $3+\epsilon$ 
            & 
            & $\tildeBigO{n^{0.397}+n^{x/2}}$ \cite{kuhn2020computing} 
            \\  
        % $n^{x}$-SSP 
            & weighted
            & $7+\epsilon$ 
            & 
            & $\tildeBigO{n^{1/3}/\epsilon+n^{x/2}}$ \cite{kuhn2020computing}
            \\  
            
        \hline
        \multirow{2}{*}{eccentricities}
            & unweighted
            & $1+\epsilon$
            & $\tildeBigO{n^{1/3}/\epsilon}$
            & 
            \\
        % eccentricities
            & weighted
            & $3$
            & $\tildeBigO{n^{1/3}}$
            & 
            \\
            
        \hline
        \multirow{7}{*}{diameter}
            & unweighted 
            & exact
            & 
            & $\tildeBigOmega{n^{1/3}}$ \cite{kuhn2020computing}
            \\
        % diameter  
            & unweighted 
            & $1+\epsilon$
            & $\tildeBigO{n^{1/3}/\epsilon}$
            & $\tildeBigO{n^{0.397}/\epsilon}$ \cite{kuhn2020computing}
            \\
        % diameter  
            & unweighted 
            & $3/2+\epsilon$
            & 
            & $\tildeBigO{n^{1/3}/\epsilon}$ \cite{kuhn2020computing}
            \\
            
        % \hline
        % diameter
            & weighted 
            & $2-\epsilon$
            & 
            & $\tildeBigOmega{n^{1/3}}$ \cite{kuhn2020computing}
            \\

        % diameter
            & weighted 
            & $2$
            & $\tildeBigO{n^{1/3}}$
            & $\tildeBigO{n^{2/5}}$ \cite{kuhn2020computing}
            \\
        % diameter
            & weighted 
            & $2+\epsilon$
            & 
            & $\tildeBigO{n^{1/3}\cdot \epsilon^{-6}}$
            \cite{AHKSS20}
            \\
        % diameter 
            & weighted 
            & $2 \cdot (1/\epsilon)^{O(1/\epsilon)}$ %\v{27} 
            & 
            & $n^{\epsilon}$\cite{AHKSS20}%\v{$n^{1/3}$}
            \\
        \hline
    \end{tabular}
    \caption{
    Comparison of our results. 
    $SPD$ is the length of the shortest path diameter.  
    The results for $n^{x}$-RSSP, and weighed diameter approximation upper bounds from previous works are implicit in~\cite{AHKSS20,kuhn2020computing}.
    Our upper bound for $n^{2/3}$-RSSP is tight up to poly-logarithmic factors due to our lower bound. Our approximations for $n^{x}$-SSP are also tight up to poly-logarithmic factors for $x\geq2/3$, due to \cite{kuhn2020computing}.
    }
    \label{table:results}
\end{table}

\subsection{Related Work}
\textbf{Hybrid Models.} 
The \hybrid network model was studied in  \cite{AHKSS20,kuhn2020computing,feldmann2020fast}. In \cite{feldmann2020fast}, distance results are obtained in one of the harsher variants of the model, where the local edges are restricted to have $\log{n}$ bandwidth. However, these apply only to extremely sparse graphs of at most $n+O(n^{1/3})$ edges and cactus graphs. In \cite{YungALS90,Gmyr2017DistributedMO}, slightly different models of hybrid nature are studied.

Augustine et al. \cite{AGGHKL19} proposed the \ncc model, which is similar to the \clique model, but each node has $\log{n}$ bandwidth. This model is also a special case of the generalised \hybrid model \cite{AHKSS20} without local edges. This allows one to use the results from the \ncc model in the \hybrid model without modifications. 

\textbf{Distributed Distance Computations.} 
Distance related problems have been extensively studied in many distributed models. For example, in the \congest model, there is a long line of research on APSP \cite{Peleg2012DistributedAF,Lenzen2013EfficientDS,Lenzen2015FastPD,Agarwal2018ADD,Elkin2017DistributedES,Bernstein2019DistributedEW,Agarwal2019DistributedWA,Ramachandran20FasterDA} which culminated in
tight, up to polylogarithmic factors, $\tildeBigO{n}$ round exact weighted APSP randomized algorithm of Bernstein and Nanongkai \cite{Bernstein2019DistributedEW} and a $\tildeBigO{n^{4/3}}$ round deterministic algorithm of Agarwal and Ramachandran \cite{Ramachandran20FasterDA}. \cite{Peleg2012DistributedAF,Lenzen2013EfficientDS} develop an $\tildeBigO{n}$ round algorithm, optimal up to polylogarithmic factors, for unweighted APSP.
The study of approximate SSSP algorithms was the focus of many recent paper \cite{Sarma2011DistributedVA,Becker2017NearOptimalAS,Henzinger2019ADA} and lately Becker et al. \cite{Becker2017NearOptimalAS} showed the solution which is close to the lower bound of Das Sarma et al. \cite{Sarma2011DistributedVA}. In case of exact SSSP, after recent works  \cite{Elkin2017DistributedES,Ghaffari2018ImprovedDA,Forster2018AFD}, there still is a gap between upper and lower bounds. The diameter and eccentricities problems are studied in the \congest model in \cite{Peleg2012DistributedAF,Frischknecht2012NetworksCC,Abboud2016NearLinearLB,AnconaCDEV20}.

\sloppy{
In the \clique model, $k$-SSP, APSP and diameter are extensively studied in \cite{Nanongkai2014DistributedAA,CensorHillel2016AlgebraicMI,Gall2016FurtherAA,CensorHillel2020SparseMM} and approximate versions of the $k$-SSP and APSP problem are solved in polylogarithmic \cite{SparseMMPODC2019} and even polyloglogaritmic \cite{Dory2020ExponentiallyFS} time. In the more restricted \bcc model, in which each message a node sends in a round is the same for all recipients, distance computations are researched by \cite{Holzer2015ApproximationOD,Becker2017NearOptimalAS}.
}

\section{Preliminaries}\label{sec:prelim}
We provide here some definitions and claims that are critical for reading the main part of the paper. \Cref{app:prelim} contains additional definitions and basic claims.
We use the following variant of the \hybrid model, introduced in \cite{AHKSS20}.

\begin{definition}[\hybrid Model]
In the \hybrid model, a synchronous network of $n$ nodes with identifiers in \interval{n}, is given by a graph $G=(V,E)$. In each round, every node can send and receive $\lambda$ messages of $O(\log n)$ bits to/from each of its neighbors (over \emph{local edges}) and an additional $\gamma$ messages in total to/from any other nodes in the network (over \emph{global edges}). If in some round more than $\gamma$ messages are sent via global edges to/from a node, only $\gamma$ messages selected adversarially are delivered.

\end{definition}

We follow the previous work of~\cite{AHKSS20,kuhn2020computing} and consider $\lambda=\infty,\gamma=O(\log{n})$. Notice that the \hybrid model can also capture the classic \local \footnote{The \local and \congest models are  synchronous distributed models where every two neighbors in the graph can exchange messages of unlimited size or of $O(\log n)$ bits, respectively, in each round.} model, with $\lambda=\infty,\gamma=0$, the classic \congest model, with $\lambda=O(1),\gamma=0$, the \clique model, with $\lambda=O(1),\gamma=0$ and $G$ being a clique, the \clique + Lenzen's Routing with $\lambda=0,\gamma=n$ and the \ncc model~\cite{AGGHKL19}, with $\lambda=0,\gamma=O(\log(n))$.

Many of our results hold for weighted graphs $G=(V, E, w)$. We assume an edge weight is given by a function $w\colon E\mapsto \set{1, 2, \dots, W}$ for a $W$ which is polynomial in $n$. When we \emph{send} an edge as part of a message in any algorithm, we assume it is sent along with its weight.

\subsection{Notations and Problem Definitions}

We use the following definitions related to graphs.
Given a graph $G=(V, E)$ and a pair of nodes $u, v\in V$, we denote by $hop(u,v)$ the hop distance between $u$ and $v$, by $N_G^h(v)$ the $h$-hop neighborhood of $v$, by $d_G^h(u, v)$ the weight of the lightest path between $u$ and $v$ of at most $h$-hops, and if there is no path of at most $h$-hops then $d_G^h(u, v)=\infty$. In the special case of $h=1$, we denote by $N_G(v)$ the neighbors of $v$ and in the special case of $h=\infty$, we denote by $d_G(u, v)$ the weight of the lightest path between $u$ and $v$. We also denote by $\deg_G{(v)}$ the degree of $v$ in $G$. Whenever it is clear from the context we drop the subscript of $G$ and just write $N$, $N^h$, $d$, $d^h$ or $\deg(v)$.

We define the following problems in the \hybrid model. 

\begin{definition}[$k$-Source Shortest Paths (k-SSP)]
    Given a graph $G = (V, E)$, and a set $S\subseteq V$ of $k$ sources. Every $u\in V$ is required to learn the distance $d_G(u, s)$ for each $s\in S$. The case where $k=1$, is called the \emph{single source shortest paths} problem (SSSP).
\end{definition}

\begin{definition}[$n^x$-Random Sources Shortest Path ($n^x$-RSSP)]\label{def:mrssp}
    Given a graph $G = (V, E)$, and a set $M\subseteq V$ of sources, such that each $v \in V$ is sampled independently with probability $n^{x-1}$ to be in $M$. Every $u\in V$ is required to learn the distance $d_G(u, s)$ for each $s\in M$.
\end{definition}

In the approximate versions of these problems, each $u\in V$ is required to learn an $\left(\alpha,\beta\right)$-approximate distance $\widetilde{d}(u, v)$ which satisfies $d(u, v)\leq \widetilde{d}(u, v) \leq \alpha \cdot d(u, v) + \beta$, and in case $\beta=0$, $\widetilde{d}(u, v)$ is called an $\alpha$-approximate distance.

\begin{definition}[Eccentricity and diameter]
    Given a graph $G=(V, E)$ and node $v\in V$, the \emph{eccentricity} of $v$ is the farthest shortest path distance from $v$, i.e., $ecc(v)=\max_{u\in V}{d(v, u)}$ and the diameter $D=\max_{v\in V}\set{ecc(v)}$ is the maximum eccentricity. An $\alpha$-approximation of all eccentricities is a function $\widetilde{ecc}(v)$ which satisfies $ecc(v)/\alpha\leq \widetilde{ecc}(v)\leq ecc(v)$ for all nodes $v$. An $\alpha$-approximation of the diameter is a value $\widetilde{D}$ which satisfies $D/\alpha\leq\widetilde{D}\leq D$.
\end{definition}

\subsection{Skeleton Graphs}

In a nutshell, given a graph $G = (V, E)$, a skeleton graph  $S_x = (M, E_S)$, for some constant $0 < x < 1$, is generated by letting every node in $V$ independently join $M$ with probability $n^{x-1}$. Two nodes in $M$ have an edge in $E_S$ if there exists a path between them in $G$ of at most $h = \tilde O(n^{1-x})$ hops. This graph \whp satisfies many useful properties in terms of distance computation, which for simplicity of presentation we add to its definition, provided below. A crucial property is that for any two nodes, if the shortest path between them in $G$ has more than $h$ hops, then there exists a shortest path between them in $G$ on which every roughly $h$ nodes there is a node from $M$ (all such skeleton properties hold \whp).

\begin{restatable}[Skeleton Graph, Combined Definition of~\cite{AHKSS20,kuhn2020computing}]{definition}{Skl}\label{def:skeleton}
Given a graph $G=(V, E)$ and a value $0 < x < 1$, a graph $S_x=(M, E_S)$ is called a skeleton graph in $G$, if all of the following hold.
\begin{enumerate}\setlength{\itemsep}{1mm}
    \item{Each $v\in V$ is included to $M$ independently with probability $n^{x-1}$.}
    
    \item{$\{v, u\} \in E_S$ if and only if there is a path of at most $h=\Tilde{\Theta}{(n^{1-x})}$ edges between $v, u$ in $G$.\label{itm:skeleton:edge}}
    
    \item Every node $v \in M$ knows all its incident edges in $E_S$.
    
    \item{$S_x$ is connected. \label{itm:skeleton:connected}}
    
    \item{For any two nodes $v, v'\in M$, $d_S(v, v')=d_G(v,v')$. \label{itm:skeleton:distance}}
    
    \item{For any two nodes $u,v\in V$ with $hop(u, v)\geq h$, there is at least one shortest path $P$ from $u$ to $v$ in $G$, such that any sub-path $Q$ of $P$ with at least $h$ nodes contains a node $w\in M$.\label{itm:skeleton:sp}}
    
    \item{$|M| = \tilde O(n^{x})$. 
    \label{itm:skeleton:size}}
\end{enumerate}
\end{restatable}

~\\The following claim summarizes what is proven in \cite{AHKSS20} regarding the construction of a skeleton graph from a set of random marked nodes, \whp

\begin{restatable}[Skeleton from Random Nodes]{claim}{SklFromRand}\label{claim:skeletonOnSampled}
Given a graph $G=(V,E)$, a value $0 < x < 1$, and a set of nodes $M$ marked independently with probability $n^{x-1}$, 
there is an algorithm in the \hybrid model which
constructs a skeleton graph $S_x=(M, E_S)$ in \tildeBigO{n^{1-x}} rounds \whp
If also given a single node $s \in V$, it is possible to construct $S_x=(M\cup\set{s}, E_S)$, without damaging the properties of $S_x$.
\end{restatable}

~\\We extract the following basic claim, used in the proof of \cite[Theorem 2.7]{AHKSS20} for a ($1+\epsilon$)-approximation for SSSP, and slightly extend it to use for multiple source problem and arbitrary approximation factors. It states that given a skeleton graph and a set of sources, if every skeleton node knows any approximation to its distance from every source, then it is possible to efficiently reach a state where every node in the graph knows the approximation for its own distance from any of the sources. We give the proof for the sake of self-containment in \cref{app:prelim}. The idea is that each node locally explores its $\tildeBigO{n^{1-x}}$ neighborhood and identifies for each source the best skeleton node in its neighborhood to go through. 

\begin{restatable}[Extend Distances]{claim}{ExtDist}\label{thr:computeMSSPFromMSSPOnSkeleton} \cite[Theorem 2.7]{AHKSS20}
    Let $G=(V, E)$, let $S_x=(M, E_S)$ be a skeleton graph, and let $V'\subseteq V$ be the set of source nodes. If for each source node $s\in V'$, each skeleton node $v\in M$ knows the $\left(\alpha, \beta\right)$-approximate distance $\tilde{d}\left(v, s\right)$ such that $d(v, s)\leq \tilde{d}(v, s)\leq \alpha d(v, s) + \beta$, then each node $u\in V$ can compute for all source nodes $s\in V'$, a value $\tilde{d}(u, s)$ such that $d(u, s)\leq \tilde{d}(u, s)\leq \alpha d(u, s)+ \beta$ in $\tildeBigO{n^{1-x}}$ rounds.
\end{restatable}

\section{Oracles in the \hybrid model}\label{sec:superCubeRootAlgorithms:tool}
This section is split into three parts. Initially, as preliminaries, we show simulations of the \local and \clique models in the \hybrid model, citing \cite{kuhn2020computing} for the \clique simulation. Then, we devote a section to each of the two new oracle models in order to introduce them and present their simulations in the \hybrid model.

\subsection{Model Simulation Preliminaries}

We will use simulations of the \local and \clique models as follows.

\begin{restatable}[\local Simulation]{lemma}{LocalSim}\label{claim:LOCALSimulation}
Given a graph $G=(V, E)$, and a skeleton graph $S_x=(M, E_S)$, it is possible to simulate one round of the \local model over $S_x$ within $\tilde O(n^{1-x})$ rounds in $G$ in the \hybrid model. That is, within $\tilde O(n^{1-x})$ rounds in $G$ in the \hybrid model, any two adjacent nodes in $S_x$ can communicate any amount of data between each other.
\end{restatable}

~\\The proof follows trivially due to the definition of the \hybrid model and Property~\ref{itm:skeleton:edge} in the definition of a skeleton graph $S_x$, since in $S_x$ two skeleton nodes are connected if they are within $\tilde \Theta(n^{1-x})$ hops in the original graph $G$. Thus, one round of the \local model  over $S_x$ is obtained in the \hybrid network in $\tilde O(n^{1-x})$ rounds, by having neighboring skeleton nodes communicate through the local edges.

\begin{restatable}[\clique Simulation]{lemma}{CCSim}\label{thr:CliqueInHybridSimulation} {\cite[Corollary 4.1.]{kuhn2020computing}}
    Given a graph $G=(V, E)$, and a skeleton graph $S_x=(M, E_S)$, for some constant $0 < x < 1$, it is possible to simulate one round of the \clique model over $S_x$ in \tildeBigO{n^{2x-1}+n^{\frac{x}{2}}} rounds of the \hybrid model on $G$, \whp That is, within \tildeBigO{n^{2x-1}+n^{\frac{x}{2}}} rounds of the \hybrid model on $G$, \whp, every node $v \in M$ can, for each node $u \in M$, each send a unique $O(\log n)$ bit message to $u$.
\end{restatable}

\subsection{Simulating the \oracle Model }\label{sec:superCubeRootAlgorithms:tools:singleOracle}

Here, we define the \oracle model and then show how to efficiently simulate it over a skeleton graph in the \hybrid model.

\begin{definition}[\oracle Model]
\label{def:oracle}
    In the \oracle model over a network $G$, there exists one \emph{oracle} node $\ell$, which in every round can send to and receive from every node $v$ a number of $\bigO{\log n}$-bit messages that is equal to the degree of $v$ in $G$. 
\end{definition}

\begin{theorem}[\oracle Simulation]

\label{thr:oracleInHybridSimulation}
    Given a graph $G=(V, E)$, for every constant 
    $0 < x < 1$, there is an algorithm which simulates one round of the \oracle model, on a skeleton graph $S_x=(M, E_S)$, in $\tildeBigO{n^{1-x} + n^{2x-1}}$ rounds of the \hybrid model on $G$, \whp 

\end{theorem}
\begin{proof}

    We prove the claim by showing how to simulate a round of the \oracle model in $O(1)$ rounds of the \clique model and 1 round of the \local model. Then, invoking the simulations of~\Cref{claim:LOCALSimulation,thr:CliqueInHybridSimulation}, gives the desired round complexity in the \hybrid model.

We show how to send messages to the oracle, and the receiving part is symmetric. The pseudocode is given by \cref{alg:oracleInHybridSimulation}.
        First, each $v\in M$ broadcasts its degree in $S_x$ using one round of the \clique model (\cref{line:oracleInHybridSimulation:broadcastDegree}) and selects as an oracle $\ell$ the node with largest degree in $S_x$, breaking ties by identifier (\cref{line:oracleInHybridSimulation:selectOracle}). Then, the identifiers of the neighbors of $\ell$ are broadcast using one round of the \clique model (\cref{line:oracleInHybridSimulation:broadcastNeighbors}). The actual messages are sent to these neighbors instead of to $\ell$ itself (\cref{line:oracleInHybridSimulation:send}) and  $\ell$ learns all these messages in $1$ round of the \local model in \cref{line:oracleInHybridSimulation:receive}.
    
    Clearly, all the nodes select the same oracle $\ell$ (\cref{line:oracleInHybridSimulation:selectOracle}). Due to the definition of the \oracle model, each node $v\in M$ has $\deg_{S_x}(v)$ messages to send, and since $\deg_{S_x}(\ell)\geq \deg_{S_x}(v)$, there are enough neighbors of $\ell$ to receive one message from $v$ per neighbor, which is why \cref{line:oracleInHybridSimulation:send} can work. 
\end{proof}
   \begin{algorithm}
        \caption{Simulating the \oracle model in the \clique with \local.}
        \label{alg:oracleInHybridSimulation}
        
        \clique model: $v\in M$ broadcasts $\deg_{S_x}(v)$ \label{line:oracleInHybridSimulation:broadcastDegree} 
        
        Select an oracle $\ell\gets \arg\max_{v\in M}\set{\left(\deg_{S_x}(v), v\right)}$ \label{line:oracleInHybridSimulation:selectOracle} 
        
        \clique model: $v\in M$ broadcasts if it is a neighbor of $\ell$  \label{line:oracleInHybridSimulation:broadcastNeighbors} 
        
        \clique model: $v\in M$ sends $i$-th message to $i$-th neighbor of $\ell$ for each $i$ \label{line:oracleInHybridSimulation:send}
        
        \local model: $\ell$ collects the messages from its neighbors \label{line:oracleInHybridSimulation:receive} 
    \end{algorithm}

\subsection{Simulating the \tiered Model}\label{app:tiered}
We further enhance our \oracle model and define the \tiered model, where, roughly speaking, all nodes in parallel can learn all the edges adjacent to nodes with degrees in lower degree buckets. To simulate the stronger \tiered model over a skeleton graph in the \hybrid model, we need additional insights. Here, we use the fact that when we scatter messages independently at random, denser neighborhoods are more likely to receive a given message than sparse neighborhoods. In other words, while for simulating the \oracle model we used the \local round only to concentrate information in a single node $\ell$, here we exploit the information that \emph{each} node can gather from its neighborhood.
 
\begin{definition}[\tiered Model]
\label{def:tiered}
    In the \tiered model over a network $G$, in every round, suppose each node $v$ has a set of $\bigO{\log n}$-bit messages $M_v$ of size $|M_v| = \deg(v)$, then each node $u$ can receive all messages in $M_v$ for every $v$ such that $\deg(u)\geq \deg(v)/2$.
\end{definition}

To simulate the \tiered model, we first prove the following model-independent tool.

\begin{restatable}[Sampled neighbors]{lemma}{DegSim}\label{lemma:autoSimulator}
    Given is a graph $G=(V, E)$. For a value $c\leq n$, there is a value $x=\tildeBigO{n / c}$ such that the following holds \whp:
    Let $V'\subseteq V$ be a subset of $\size{V'}=x$ nodes sampled uniformly at random from $M$.
    Then each node $u\in V$ with $\deg{(u)}\geq c$ has a neighbor in $V'$.

\end{restatable}
\begin{proof}
	For some node $u\in V$, the probability of not having a neighbor sampled to the set  $V'$ is $(1 - \deg{(u)}/n)^{x}\leq e^{-x\cdot\det{(u)}/n}\leq e^{x\cdot c / n}$. Thus, there exists $x=\tildeBigO{n/c}$ such that node $u$ has a neighbor in the set $V'$, \whp.
\end{proof}

Finally, we show how to simulate the \tiered model over the skeleton graph in the \hybrid model.

\begin{theorem}[\tiered Simulation]\label{thr:tieredInHybridSimulation}
    Given a graph $G=(V, E)$, for every constant 
    $0 < x < 1$, there is an algorithm which simulates one round of the \tiered model, on a skeleton graph $S_x=(M, E_S)$, in $\tildeBigO{n^{1-x} + n^{2x-1}}$ rounds of the \hybrid model on $G$, \whp
\end{theorem}

\begin{proof}
     We prove the claim by reducing one round of the \tiered model to $\tildeBigO{1}$ rounds of the \clique model followed by a round of the \local model on the skeleton graph $S_x$. By \cref{claim:LOCALSimulation,thr:CliqueInHybridSimulation}, we obtain that the resulting round complexity is $\tildeBigO{n^{1-x} + n^{2x-1}}$.
     
	For each $v \in M$, let $M_v$ be the set of messages, of size $|M_v| = \deg_{S_x}{(v)}$, which $v$ desires to broadcast. For each message in $M_v$ node $v\in M$ samples uniformly at random $x=\tildeBigO{2\cdot n / \deg_{S_x}{(v)}}$ nodes of $M$ and sends the message to those nodes. As each node sends and receives $\tildeBigO{\size{M}}$ messages this can be done using with the well known routing theorem of Lenzen~\cite[Theorem 3.7]{Lenzen} by simulating $\tildeBigO{1}$ rounds of the \clique model. Alternatively, this can be done in the same round complexity  by applying the algorithm for \emph{token routing} \cite[Theorem 2.2]{kuhn2020computing}. Afterwards, we simulate one round of the \local model over $S_x$ for each node to learn tokens received by its neighbors in $S$. Due to \fullref{lemma:autoSimulator}, each node $u\in V$ learns messages from each $v$ such that $\deg_{S_x}{(u)}\geq \deg_{S_x}{(v)} / 2$ \whp
\end{proof}
\section{Shortest Paths Algorithms}\label{sec:superCubeRootAlgorithms}

\subsection{Warm-Up: Exact SSSP}
As a warm-up, we show how to compute exact SSSP in the \oracle model, and then we simulate this on a skeleton graph in the \hybrid model in order to get exact SSSP in the \hybrid model within $\tilde{O}(n^{1/3})$ rounds. We note that later, in \Cref{sec:Nthird}, we obtain this complexity for exact distances from a much larger set, of $O(n^{1/3})$ sources.

\begin{lemma}[Exact SSSP in the \oracle Model]\label{lemma:oracle:exactSSSP}
	There is a \emph{deterministic} algorithm in the \oracle model that given a weighted graph $G=(V, E)$ and source $s\in V$ solves exact SSSP in \bigO{1} rounds.
\end{lemma}

\begin{proof}
    
    Let $s\in V$ be the source node. 
    We solve the problem in two communication rounds. In the first round, 
    oracle $\ell$ learns all of $E$ by receiving from each node $v$ its adjacent edges. Afterwards, oracle $\ell$, given all the edges in the graph $G$, locally computes the distance from $s$ to every other node. 
    In the second round, 
    oracle $\ell$ sends for each $v \in V$ the value $d(s, v)$. It is clear that the algorithm computes SSSP from $s \in V$, and that it takes two rounds in the \oracle model.
\end{proof}

\exactSSSP*

\begin{proof}

    Let $s$ be the source node, and let $x=2/3$. 
    We start by constructing a skeleton graph $S_x=(M, E_S)$, 
    by sampling nodes with probability $n^{-1/3}$ and using \fullref{claim:skeletonOnSampled}. Then, we simulate the algorithm given in \Cref{lemma:oracle:exactSSSP} 
    in the \oracle model, which
    computes the distance $d_S(s, v)$ from $s$ to each node $v\in M$. 
    By Property~\ref{itm:skeleton:distance} of the skeleton graph, for every $v\in M$, it holds that $d_S(s, v)=d_G(s, v)$. To extend this and compute the distance from $s$ for each node $v\in V$, we apply
    \fullref{thr:computeMSSPFromMSSPOnSkeleton}.

    Constructing the skeleton graph takes $\bigO{h}=\tildeBigO{n^{1/3}}$ rounds \whp, by \fullref{claim:skeletonOnSampled}.
    Simulating the algorithm from \Cref{lemma:oracle:exactSSSP} completes in $\tildeBigO{n^{1/3}}$ rounds \whp by \fullref{thr:oracleInHybridSimulation}. Applying \fullref{thr:computeMSSPFromMSSPOnSkeleton} takes $\tildeBigO{n^{1/3}}$ rounds. Therefore, overall, the execution of the algorithm completes in $\tildeBigO{n^{1/3}}$ rounds \whp
\end{proof}

\subsection{Exact \texorpdfstring{$n^{x}$-}{M}RSSP}\label{subsec:exactRMSSP}

Recall that in \fullref{def:mrssp}, we are given set of roughly $n^{x}$ sources sampled independently with probability $n^{x-1}$, and we need for each node to compute its distance to each source. We do so by constructing a skeleton graph $S_x$ from the random sources. We show that using one round of the \tiered model, and $O(\log n)$ rounds of the \clique model, one can solve APSP over $S_x$. To do so, we split the nodes of the graph into $\ceil*{\log n}$ tiers by degree and compute APSP by proceeding tier after tier and computing distances from current tier to all the tiers below.

\begin{restatable}[APSP in \clique with \tiered]{lemma}{cliqueWithTieredExactRAPSP}\label{thr:cliqueWithTiered:exactRAPSP}
There is a \emph{deterministic} algorithm which, given a weighted graph $G=(V, E)$, solves exact APSP on $G$ using \bigO{\log{\size{V}}} rounds of the \clique model and one round of the \tiered model.
\end{restatable}
\begin{proof}

    The pseudocode for the algorithm appears in \cref{alg:cliqueTieredexactRAPSP}. We partition the nodes $V$ by their degrees into $\ceil*{\log{\size{V}}}$ tiers, $T_j=\Set{v\in V\colon 2^j\leq \deg(v)<2^{j+1}}$ for $0 \leq j < \ceil*{\log{\size{V}}}$. Denote by $T_{>i}=\bigcup_{k> i}{T_k}$ the nodes in all tiers $k> i$ and by $T_{\leq i}=\bigcup_{k\leq i}{T_k}$ the nodes in all tiers $k\leq i$. Similarly, define $T_{\geq i}$ and $T_{< i}$. Denote by $d_{\leq i}(u, v)$ the weight of the shortest path between $u$ and $v$ that uses only edges adjacent to at least one node in $T_{\leq i}$.

    \begin{algorithm}
        \tiered model: each $v\in T_i$ broadcasts to $u\in T_{\geq i}$ its incident edges  \label{line:sendToUpperTiers}  
        
        \For{$i=\ceil*{\log \size{V}} - 1$ downto $0$ \label{loop:tiers}}{
            For each node $u\in T_{\leq i}$, each node $v\in T_{i}$ computes $\tilde{d}(v, u)\gets \min{\set{d_{\leq i}(v, u), \min_{w\in T_{>i}}{\set{\tilde{d}({v, w}) + d_{\leq i}(w, u)}}}}$
            \label{line:computeMSSP}

            \clique model: $v\in T_{i}$ sends to $u\in T_{\leq i}$, the value $\tilde{d}(v, u)$   \label{line:informDistanceToItself} 
        }
        
        \caption{\textbf{Exact-APSP}: Computes exact APSP using the \clique and \tiered models}
        \label{alg:cliqueTieredexactRAPSP}
    \end{algorithm}

    The outline of our algorithm is as follows. We start by having each node $v\in T_i$ broadcast its incident edges to all the nodes in tiers greater than or equal to its own, that is, to all $u\in T_{\geq i}$, using one round of the \tiered model (\cref{line:sendToUpperTiers}). Afterwards, in the loop in~\cref{loop:tiers}, we compute the solution tier by tier, starting from the topmost tier, which contains nodes knowing all the edges in the graph. While processing the $i$-th tier, every node $v \in T_i$ already knows its distance to every node in $T_{>i}$, and so computes its distances to every node $u \in T_{\leq i}$. A shortest path between such $v$ and $u$ can either pass through edges which are all known to $v$, or be broken into a subpath from $v$ to some node $w \in T_{>i}$ and then a path from $w$ to $u$ which is known to $v$. Thus, we compute the distance from $v\in T_i$ to the nodes $T_{\leq i}$ (\cref{line:computeMSSP}). On \cref{line:informDistanceToItself}, node $v\in T_i$, which knows for each node $u\in T_{\leq i}$ the distance to $u$, sends it to $u$.
    
    For each $u, v \in V$, \cref{alg:cliqueTieredexactRAPSP} outputs a value $\tilde{d}(u, v)$. We show that it is the correct distance in $G$, that is $\tilde{d}(u, v) = d(u, v)$. 
    
    One round of the \tiered model suffices for ensuring that for each tier, $T_i$, every node $v \in T_i$ knows all the edges incident to all the nodes $u \in T_{\leq i}$. Let $v\in T_i$, and $u\in T_j$ such that $i\geq j$, observe that it holds that $\deg(v)\geq 2^{i}\geq \frac{1}{2}2^j=\frac{1}{2}\deg(u)$, and therefore after \cref{line:sendToUpperTiers} node $v$ knows the edges incident to $u$. Thus, each node $v \in T_i$ knows enough information to compute the function $d_{\leq i}$, which is the distance function in $G$ limited to edges incident to nodes in $T_{\leq i}$.
    
    By induction on tier index $i$, we show that after iteration $i$ of the loop in~\cref{loop:tiers} all the nodes in $V$ know the exact distances to all nodes in tiers $T_{\geq i}$. 

    Base case: In iteration $i=\ceil*{\log{\size{V}}} - 1$, node $v\in T_{\ceil*{\log{\size{V}}} - 1}$ (if exists) in the topmost tier knows about all the edges in $E$ since it knows about all edges incident to nodes $T_{\leq \ceil*{\log{\size{V}}} - 1} = V$. Thus, $v$ can compute the solution to the entire APSP on $G$, since $d_{\leq \ceil*{\log{\size{V}}} - 1} = d$. Since the set  $$\set{d({v, w}) + d_{\leq \ceil*{\log{\size{V}}} - 1}(w, u)}_{w\in T_{>\ceil*{\log{\size{V}}} - 1}}$$ is empty, we get that $\tilde{d}(v, u)=d_{\leq \ceil*{\log{\size{V}}} - 1}(v, u)=d(v, u)$. That is, node $v\in T_{\ceil*{\log{\size{V}}} - 1}$ computes for each other node $u\in V$ its weighted distance $d(v, u)$ and sends it to $u$ on \cref{line:informDistanceToItself}. 
    
    Induction Step: In iteration $i<\ceil*{\log{\size{V}}} - 1$, consider $v\in T_i$ and $u\in T_{\leq i}$, and let $P$ be a shortest path between them. Recall that node $v$ can locally compute $d_{\leq i}$, and thus knows the value $d_{\leq i}(v, u)$ and for each $w \in T_{>i}$, it knows the value $d_{\leq i}(w, u)$. Further, for each $w\in T_{>i}$, the value $\tilde{d}(v, w)=d(v, w)$ is known to $v$ from one of the previous iterations of the
    loop in~\cref{loop:tiers}, by the induction assumption. All values in the set $\set{\tilde{d}({v, w}) + d_{\leq i}(w, u)}_{w\in T_{>i}}\cup\set{d_{\leq i}(v, u)}$ are either infinite or correspond to some (not necessary simple) path from $v$ to $u$, thus $\tilde{d}(v, u)\geq d(v, u)$. To show that $\tilde{d}(v, u)\leq d(v, u)$, we consider two cases. If $P$ does not contain nodes from $T_{>i}$, then $d_{\leq i}(v, u) = d(v, u)$ is the length of $P$. Otherwise, let $w'\in T_{>i}$ be the last node on $P$ (closest to $u$) which belongs to $T_{>i}$. By the induction hypothesis, $v$ knows $\tilde{d}(v, w')=d(v, w')$. Moreover, the subpath from $w'$ to $u$ only contains edges with at least one endpoint incident to node in $T_{\leq i}$, thus $d_{\leq i}(w', u) = d(w', u)$. For this node $w'$ the value $\set{\tilde{d}({v, w'}) + d_{\leq i}(w', u)}$ belongs to ${\set{\tilde{d}({v, w}) + d_{\leq i}(w, u)}}_{w\in T_{>i}}$. Thus, in both cases the computed $\tilde{d}(v, u)$ is at most the weighted length of $P$. Hence, $\tilde{d}(v, u)=d(v, u)$. On \cref{line:informDistanceToItself}, node $v$ informs $u$ about the correct $d(v, u)$, which completes the induction proof.

    \cref{line:sendToUpperTiers,line:informDistanceToItself} each take a single round of the \tiered model and the \clique model, respectively, and thus the execution of the entire algorithm takes \bigO{\log{\size{V}}} rounds of the \clique model and one round of the \tiered model. 
\end{proof}

By simulating the algorithm given in \fullref{thr:cliqueWithTiered:exactRAPSP} using \fullref{thr:tieredInHybridSimulation,thr:CliqueInHybridSimulation}, we get exact APSP over the skeleton graph, as follows.

\begin{restatable}[Exact APSP on Skeleton Graph]{corollary}{exactRAPSP}\label{thr:exactRAPSP}
For any constant $0 < x < 1$, there is an algorithm in the \hybrid model that computes an exact weighted APSP on a skeleton graph $S_{x}=(M, E_S)$, in \tildeBigO{n^{1-x}+n^{2x-1}} rounds \whp
\end{restatable}

Finally, we extend the result to $n^{x}$-RSSP on $G$, by having each node in the graph learn the information stored in the skeletons in its $\tildeBigO{n^{1-x}}$ neighborhood. 

\exactRMSSP*
\begin{proof}
    Primarily, assume that $x\geq \frac{2}{3}$. Otherwise, we add each node outside of $M$ with probability $(n^{-1/3} - n^{x - 1})/(1 - n^{x - 1})$ into the set $M$. Thus, each node has probability exactly $(n^{x-1}\cdot 1)+(1 - n^{x - 1})\cdot(n^{-1/3} - n^{x - 1})/(1 - n^{x - 1})=n^{-1/3}$ to be sampled into $M$, ensuring $x = 2/3$. We use \fullref{claim:skeletonOnSampled} to build a skeleton graph $S_x=(M, E_S)$ in \tildeBigO{n^{1/3}} rounds \whp Then, we compute exact APSP on the skeleton graph using \fullref{thr:exactRAPSP} in \tildeBigO{n^{1/3}+n^{2x-1}} rounds \whp By Property~\ref{itm:skeleton:distance} of the skeleton graph, for each $v, u\in M$ it holds that $d_S(v, u)=d(v, u)$, where $d_S(v, u)$ is the distance in the skeleton graph. So, we apply \fullref{thr:computeMSSPFromMSSPOnSkeleton} with $\alpha=1,\beta=0$ and set of sources $V'=M$, to compute an exact weighted shortest paths distances, from $M$ to all of $V$, in additional \tildeBigO{n^{1/3}} rounds \whp 
\end{proof}

Instantiating \Cref{thr:exactRMSSP} with $x=2/3$ gives $n^{2/3}$-RSSP in $\tildeTheta{n^{1/3}}$ rounds \whp, which is tight due to our lower bound given in \fullref{thr:lowerBound}. We extensively use our $n^{2/3}$-RSSP algorithm for our following results.

\subsection{Exact \texorpdfstring{$n^{1/3}$-}{M}SSP}\label{sec:Nthird}
We now present an improvement over the warm-up exact SSSP algorithm which we showed previously, by providing an algorithm for exact shortest paths from a \emph{given} set of $n^{1/3}$ nodes ($n^{1/3}$-SSP) in \tildeBigO{n^{1/3}} rounds. To do so, we create a skeleton graph and use our algorithm for $n^{2/3}$-RSSP algorithm to compute exact distances from the skeleton nodes to the entire graph. Then, we adapt the behavior of the source nodes depending on the number of skeleton nodes in their neighborhood (which is proportional to the density of the neighborhoods). That is, nodes in sparse neighborhoods can broadcast the distances from themselves to all the skeleton nodes which they see surrounding them, while a node in dense neighborhoods can take over a skeleton node surrounding it and use it as a proxy to communicate efficiently with the other skeleton nodes in the graph. We formalize this in this section, as well as refer to \fullref{thr:reassignSkeletons} which is a generic tool which performs this action of \emph{taking over} skeleton nodes as proxies.

We show the following fundamental algorithm, which allows \emph{assigning} skeletons to help other skeletons. That is, given a set of nodes $A$ where each node in $A$ sees many skeleton nodes in its neighborhood, it is possible to assign skeleton nodes to service the nodes of $A$. We use this to increase sending and receiving \emph{capacity} of the nodes of $A$. This is a key tool which we use in the proof of \fullref{thr:exactMSSP} and we believe it may be useful for additional tasks.

\begin{lemma}[Reassign Skeletons]\label{thr:reassignSkeletons}
    Given graph $G=(V, E)$, a skeleton graph $S_x=(M, E_S)$, a value $k$ which is known to all the nodes, and nodes $A\subseteq V$ such that each $u \in A$ has at least $\tildeTheta{k\cdot |A|}$ nodes $M_u \subseteq M$ in its $\tilde \Theta(n^{1-x})$ neighborhood, there is an algorithm that assigns $K_u \subseteq M_u$ nodes to $u$, where $|K_u|=\tildeBigOmega{k}$, such that each node in $M$ is assigned to at most $\tilde O(1)$ nodes in $A$. With respect to the set $A$, it is only required that every node in $G$ must know whether or not it itself is in $A$ -- that is, the entire contents of $A$ do not have to be globally known. The algorithm runs in \tildeBigO{n^{1-x}} rounds in the \hybrid model, \whp
\end{lemma}
\begin{proof}

    The pseudocode is provided by \cref{alg:reassignSkeletons}.
    
    \begin{algorithm}
        \caption{\textbf{Reassign-Skeletons}$(A, k)$}
        \label{alg:reassignSkeletons}
        
        Compute $\size{A}$ by running Aggregate-And-Broadcast \label{line:reassignSkeletons:computeSizeA} 
        
        Skeleton node $v\in M$ learns its $\tildeBigO{n^{1-x}}$-hop neighborhood \label{line:reassignSkeletons:learnNeighborhood} 
        
        Skeleton node $v\in M$ samples each $u\in A\cap N_G^{\tildeTheta{n^{1-x}}}(v)$ with probability $\frac{1}{\size{A}}$ \label{line:reassignSkeletons:sample} 
        
        Skeleton node $v\in M$ informs each sampled node $u$ about $v\in K_u$ \label{line:reassigneSkeletons:inform} 
    \end{algorithm}
    
    First, each node $w$ learns the size of the set $A$ by invoking \fullref{thr:aggregateAndBroadcast} with value $1$ if $w\in A$ and $0$ otherwise, and the summation function  (\cref{line:reassignSkeletons:computeSizeA}). Then, each skeleton node $v\in M$, learns its $\tildeTheta{n^{1-x}}$-hop neighborhood (\cref{line:reassignSkeletons:learnNeighborhood}), and in particular it learns the nodes $A_v=A\cap N_G^{\tildeTheta{n^{1-x}}}(v)$. Then, $v$ samples each $u\in A_v$ independently with probability $\frac{1}{\size{A}}$ (\cref{line:reassignSkeletons:sample}). Afterwards, $v$ informs each node $u$ it sampled on the previous stage that $v\in K_u$ (\cref{line:reassigneSkeletons:inform}).
    
    For every $v\in M$, since $\size{A_v}\leq \size{A}$, and since $v$ samples nodes from there $A_v$ independently with probability $\frac{1}{\size{A}}$, by Chernoff Bounds each $v$ assigns itself to at most \tildeBigO{1} nodes $a\in A_v$ \whp Hence, by a union bound over all skeleton nodes, each skeleton node is assigned to $\tildeBigO{1}$ nodes \whp
    
    For every $u\in A$, since it is sampled by at least $\tildeBigOmega{k\cdot \size{A}}$ skeleton nodes independently with probability $\frac{1}{\size{A}}$, by Chernoff Bounds it is sampled by $\size{K_u}=\tildeBigOmega{1}$ skeleton nodes \whp Thus, by union bound over all skeleton nodes, each $u\in A$ has $\size{K_u}=\tildeBigOmega{1}$ assigned nodes \whp 
    
    By \fullref{thr:aggregateAndBroadcast}, \cref{line:reassignSkeletons:computeSizeA} takes \tildeBigO{1} rounds \whp, and \cref{line:reassignSkeletons:learnNeighborhood,line:reassigneSkeletons:inform} take \tildeBigO{n^{1-x}} rounds, and thus the entire execution completes in \tildeBigO{n^{1-x}} rounds \whp
\end{proof}

Now we apply \fullref{thr:tokenDissemination}~or~\fullref{thr:reassignSkeletons} depending on density of each source's neighborhood and show how to compute exact $n^{1/3}$-SSP in $\tildeBigO{n^{1/3}}$ rounds. For  sources in ``sparse'' neighborhoods, in which there is a small number of skeleton nodes, we use \fullref{thr:tokenDissemination} to inform all nodes about their distances to those skeletons. For source $v$ with ``dense'' neighborhood, in which there are many skeleton nodes, we use \fullref{thr:reassignSkeletons} to get at least one skeleton node $u$ which participates in the round of the \clique model on behalf of that source and sends each other skeleton node $v'$ the distance $d(v, v')$.

\exactMSSP*
\begin{proof}
    The pseudocode for the algorithm appears in \cref{alg:exactMSSP}.
    
        \begin{algorithm}
        \caption{\textbf{Exact-$n^{1/3}$-SSP}: Computes an exact weighted $n^{1/3}$-SSP. Routine for node $u\in V$}
        \label{alg:exactMSSP}
        
        Join $M$ independently with probability $n^{-1/3}$ \label{line:exactMSSP:sample} 
        
        Compute $n^{2/3}$-RSSP from $M$ \label{line:exactMSSP:exactRMSSP} 
        
        Construct skeleton graph $S_{2/3}=(M, E_S)$ \label{line:exactMSSP:constructSkeleton} 
        
        Learn $h=\tildeTheta{n^{1/3}}$-hop neighborhood \label{line:exactMSSP:learnSkeleton} 
        
        \If{$u\in U$} {
            \eIf{$\size{N_G^{h}(u)\cap M}=\tildeBigO{n^{1/3}}$}{
                Participate in \TokenDissemination with a token $\Braket{u, v', d(v', u)}$ for each $v'\in M\cap N_G^{h}(s)$ \label{line:exactMSSP:td}  
            }{
                $K_u\leftarrow$\ReassignSkeletons$\left(\Set{u\colon \size{N_G^{h}(u)\cap M}=\tildeBigOmega{n^{1/3}}}, \tildeBigO{1}\right)$ \label{line:exactMSSP:reassign} 
                
                Send each $v\in K_u$ the values $d(u, v')$ for each $v'\in M$ \label{line:exactMSSP:sendToHelpers} 
            }

        }
        
        \If{$u\in M$} {
            In the \clique model: for each $v'\in M$ and each $v\in U$ such that $u\in K_v$ send $d(v, v')$ to $v'$ \label{line:exactMSSP:clique} 
            
            For each $s\in U$, compute $\tilde{d}(u, s)$ by \cref{eq:exactMSSP} and output it \label{line:exactMSSP:distanceToSparse}
        }
        
        Apply 
        \Cref{thr:computeMSSPFromMSSPOnSkeleton},
        which given distances from skeleton to sources $\tilde{d}\colon M \times U\mapsto \N$ extends it to distances from each nodes to sources $\tilde{d}\colon V\times U\mapsto \N$  \label{line:exactMSSP:computeMSSPFromMSSPOnSkeleton}

    \end{algorithm}

    Without loss of generality the set of nodes $U$ is globally known (it can be disseminated in $\tildeBigO{n^{1/6}}$ rounds \whp using \fullref{thr:tokenDissemination}). We build $M\subseteq V$ by marking nodes independently with probability $n^{-1/3}$ (\cref{line:exactMSSP:sample}). Then we run the algorithm from \fullref{thr:exactRMSSP} with $x=2/3$ to obtain \whp $n^{2/3}$-RSSP from the set of nodes $M$ (\cref{line:exactMSSP:exactRMSSP}), such that \whp every $u\in V$ knows its distance to every node in $M$.
    Afterwards, we apply  \fullref{claim:skeletonOnSampled} to construct a skeleton graph $S_{{2}/{3}}=(M, E_S)$ \whp 
    Then, each source learns the information in its $h$-hop neighborhood (\cref{line:exactMSSP:learnSkeleton}), for $h\in\tildeTheta{n^{1/3}}$. In particular, it counts the skeleton nodes in its $h$-hop neighborhood. 
    
    If a source finds that the number of skeleton nodes in its $h$-hop neighborhood is $\tildeBigO{n^{1/3}}$, then it participates in a token dissemination protocol (\fullref{thr:tokenDissemination}) and \whp informs all the graph about its distance to these skeleton nodes (\cref{line:exactMSSP:td}).

    Otherwise, each source $u\in U$ which finds that there are at least $\tildeBigOmega{n^{1/3}}$ skeleton nodes in its $h$-hop neighborhood, applies \fullref{thr:reassignSkeletons} with $k=\tildeBigO{1}$ and $A=\set{u\colon \size{N_G^{h}(u)\cap M}=\tildeBigOmega{n^{1/3}}}$ and receives $K_u\subseteq N_G^h(u)\cap M$, a set of $\tildeBigOmega{1}$ skeletons.
    Such a source $u\in U$ sends by local edges to $v\in K_u$ the distance $d(u, v')$ to each $v'\in M$ (\cref{line:exactMSSP:sendToHelpers}). Each skeleton node $u\in M$ sends the distances $d(s, v')$ to each $v'\in M$, for each source node $s\in U$ that it is assigned to, by simulating the \clique model. If skeleton node $u\in M$ for source $s\in U$  did not receive the distance from $s$, it computes it using \cref{eq:exactMSSP} (\cref{line:exactMSSP:distanceToSparse}) based on the information it received in \cref{line:exactMSSP:td}.
    
    \begin{align}\label{eq:exactMSSP}
        \tilde{d}(u, s)=\min\set{d^h(u, s), \min_{v'\in M}\set{d(u, v')+d^h(v', s)}}\footnotemark
    \end{align}
    
    \footnotetext{Note that difference from \cref{eq:shortestPathApproximation}. There node $u$ does not know precise distance to skeleton, but only $h$-limited distance, while here it knows the precise distance after \cref{line:exactMSSP:exactRMSSP}. Similarly with the $d^h(v', s)$ term.}

    After each skeleton knows the distance to each source, we apply \Cref{thr:computeMSSPFromMSSPOnSkeleton}
    to compute distances from sources to all the nodes.

    \begin{restatable}{lemma}{knowsAPSP}\label{lemma:exactMSSP:knowsAPSP}
        After \cref{line:exactMSSP:distanceToSparse},  each $u\in M$ knows $d(v, s)$ for each $s\in U$ \whp
    \end{restatable}

    By \cref{lemma:exactMSSP:knowsAPSP}, whose proof appears in \cref{appsec:exactMSSP}, each node in $M$ knows the distance to each node in $U$, thus by \fullref{thr:computeMSSPFromMSSPOnSkeleton} with $\alpha=1,\beta=0$ there is an algorithm to compute shortest paths distance from $U$.
    
    By \fullref{thr:exactRMSSP} with $x=\frac{2}{3}$, \cref{line:exactMSSP:exactRMSSP}  completes in \tildeBigO{n^{1/3}} rounds \whp For \cref{line:exactMSSP:constructSkeleton}, by \fullref{claim:skeletonOnSampled}, the round complexity, \whp, is \tildeBigO{n^{1/3}} as well. \cref{line:exactMSSP:learnSkeleton} completes in \tildeBigO{n^{1/3}} rounds. Since there are at most $\ell=\tildeBigO{n^{1/3}}$ tokens per source and $k=\bigOmega{n^{2/3}}$ tokens overall, \cref{line:exactMSSP:td} takes $\tildeBigO{n^{1/3}}$ rounds \whp by \fullref{thr:tokenDissemination}. By \fullref{thr:reassignSkeletons}, \cref{line:exactMSSP:reassign} takes \tildeBigO{n^{1/3}} rounds \whp All skeleton nodes are assigned to some helpers in their $\tildeBigO{n^{1/3}}$-hop neighborhood by \cref{line:exactMSSP:reassign}, so \cref{line:reassigneSkeletons:inform} takes  $\tildeBigO{n^{1/3}}$ rounds. Since each skeleton selects \tildeBigO{1} sources \whp in \cref{line:exactMSSP:reassign} by \fullref{thr:reassignSkeletons}, \cref{line:exactMSSP:clique} simulates \tildeBigO{1} rounds of the \clique model and takes \tildeBigO{n^{1/3}} rounds by \fullref{thr:CliqueInHybridSimulation} \whp  Finally by \fullref{thr:computeMSSPFromMSSPOnSkeleton}, 
    \cref{line:exactMSSP:computeMSSPFromMSSPOnSkeleton} for $x=\frac{2}{3}$ takes \tildeBigO{n^{1/3}} rounds as well. Thus, the overall execution of the algorithm takes \tildeBigO{n^{1/3}} rounds.
\end{proof}

\textbf{Remark:} We show the approximation for shortest path from $n^{x}$ given sources in \cref{app:MSSP}. In \cref{app:ecc,app:diameter} we provide eccentricities and diameter approximations, respectively. We wrap up with our lower bound for shortest paths from sources sampled i.i.d. in \cref{sec:lowerBound}.

\section*{Acknowledgements}
This project has received funding from the European Union’s Horizon 2020 research and innovation programme under grant agreement no. 755839-ERC-BANDWIDTH

The authors would like to thank Michal Dory and Yuval Efron for various helpful conversations. We also thank Fabian Kuhn for sharing a preprint of \cite{kuhn2020computing} with us.

\bibliography{References-standard}

\appendix

\section{Preliminaries -- Missing proofs}
\label{app:prelim}

We slightly extend the claim used in the proof of \cite[Theorem 2.7]{AHKSS20} -- a basic claim regarding usage of skeleton graphs for purposes of distance computations in the \hybrid model. We include a proof for completeness.

\ExtDist*
\begin{claimproof}

    First, each node $u\in V$ learns $\tilde{d}(v, s)$ for each source node $s\in V'$ and skeleton node $v\in M$ in its $h=\tildeTheta{n^{1-x}}$-hop neighborhood. 
    Then, node $u\in V$ approximates its distance to each source $s\in V'$ using \cref{eq:shortestPathApproximation}, as follows. 
    
    \begin{align}\label{eq:shortestPathApproximation}
        \tilde{d}(u, s)=\min\Set{d^h(u, s),\min_{v\in M\cap N_G^h(u)}\Set{d^h(u, v) + \tilde{d}(v, s)}}
    \end{align}

    For each node $u\in V$, skeleton node $v\in M\cap N_G^h$, and source $s\in V'$, $d^h(u, v)\geq d(u, v)$, it holds that $\tilde{d}(v, s)\geq d(v, s)$ and $d^h(u, s)\geq d(u, s)$, by triangle inequality $d(u, v)+d(v, s)\geq d(u, s)$, thus $\tilde{d}(u, s)\geq d(u, s)$. To show that $\tilde{d}(v,s)\leq \alpha \cdot d(v, s) + 2\beta$, we consider two cases. If there is a shortest path of at most $h$ hops between $s$ and $v$, then its length appears as $d^h(u, s)$ in the \cref{eq:shortestPathApproximation}. Otherwise, by the Property~\ref{itm:skeleton:sp} of the definition of the skeleton graph, there is a shortest path from $u$ to $s$ on which there is a node $v'\in M$ in one of its $h$ first notes. This means that $d(u, v')+d(v', s)=d(u, s)$ and also $d^h(u, v')=d(u, v')$ (as each sub-path of shortest path is a shortest path) and so for this $v=v'$ the corresponding element $d^h(u, v')+\tilde{d}(v', s)$ appears as an option in \cref{eq:shortestPathApproximation}. This implies that $\tilde{d}(u, s)\leq d^h(u, v') + \tilde{d}(v', s) \leq d(u, v')+ \alpha\cdot d(v', s) + \beta = \alpha\cdot (d(u, v') + d(v', s)) + \beta\leq \alpha \cdot d(u, s) + \beta$. Thus, in both cases $\tilde{d}(u, s)\leq \alpha \cdot d(u, s) + \beta$, and therefore $\tilde{d}$ is an $\left(\alpha, \beta\right)$-approximation for the shortest path.

    Learning the information in the $h$-hop neighborhood requires $h=\tildeBigO{n^{1-x}}$ rounds and the rest
    is done locally, so overall the algorithm runs in \tildeBigO{n^{1-x}} rounds of the \hybrid model.
\end{claimproof}

\section{Oracles in the \hybrid Model -- Missing proof}
\label{app:SimsInTools}

\section{Shortest Path Algorithms -- Extended}
\label{app:bigO}
A useful communication routine which we use in this section is \fullref{thr:aggregateAndBroadcast}. It was proved for a weaker \ncc model\cite[Theorem 2.2]{AGGHKL19}, and thus directly holds in the \hybrid model.

\begin{claim}[Aggregate And Broadcast]\label{thr:aggregateAndBroadcast}
Let $S$ be some ground set and $f$ be a globally known function, which maps some multiset $S'\subseteq S$ to a value $f(S')\in S$. Assume that there exists a globally known function $g$, such that for any multiset $S'$ and any partition $S_1,\cdots ,S_\ell$ of $S'$ holds $f(S)=g(f(S_1), \dots, f(S_\ell))$. Assume also that each node $u$ has an input value $val(u)\in S$. Then there is an algorithm in the \hybrid model, which runs for \bigO{\log{n}} rounds, after which each node $u$ knows $f(v_1,\dots, v_n)$.
\end{claim}

For example, we use \Cref{thr:aggregateAndBroadcast} to sum numbers that each node has and make the result globally known.

\subsection{Exact \texorpdfstring{$n^{1/3}$-}{M}SSP ~-- Missing Proof}
\label{appsec:exactMSSP}

The first tool we need for our exact {$n^{1/3}$}-SSP algorithm is a \emph{token dissemination} algorithm, given in~\cite[Theorem 2.1]{AHKSS20}.

\begin{definition}[Token Dissemination Problem]
    The problem of making $k$ distinct tokens globally known, where each token is initially known to one node, and each node initially knows at most $\ell$ tokens is called the \emph{$(k,\ell)$-Token Dissemination (TD) problem}.
\end{definition}
\begin{claim}[Token Dissemination]~\cite[Theorem 2.2]{kuhn2020computing}
There is an algorithm that solves $(k, \ell)$-TD in the \hybrid model in $\tildeBigO{\sqrt{k}+\ell}$ rounds, \whp
\label{thr:tokenDissemination}
\end{claim}

    \knowsAPSP*
    \begin{proofof}{\cref{lemma:exactMSSP:knowsAPSP}}
        Let $u\in M, s\in U$. To show the claim we split into the two cases by the number of sources in $h$-hop neighborhood of $s$. If the neighborhood is dense, meaning $\size{N_G^{h}(s)\cap M}= \tildeBigOmega{n^{1/3}}$, then $s$ knows $d(s, u)$ \whp after \cref{line:exactMSSP:exactRMSSP} by \fullref{thr:exactRMSSP} and sends it to $u$ in \cref{line:exactMSSP:clique} by \fullref{thr:reassignSkeletons,thr:CliqueInHybridSimulation} In this case \cref{eq:exactMSSP} gives the exact answer for $v'=u$ \whp
        Otherwise, the distance is computed by \cref{eq:exactMSSP}. 
        
        First, we show that node $u\in M$ on \cref{line:exactMSSP:distanceToSparse} can compute \cref{eq:exactMSSP}. From the local exploration in \cref{line:exactMSSP:learnSkeleton} $u$ learns $d^h(u, s)$. After \cref{line:exactMSSP:exactRMSSP} $u$, \whp, knows for each $v'\in M$ $d(u, v')$ and due to \fullref{thr:tokenDissemination}. In \cref{line:exactMSSP:td} it receives $d(v', s)$ \whp
        
        Now we prove that the value $\tilde{d}(u, s)$, computed by \cref{eq:exactMSSP}, \whp equals $d(v, s)$. Each element in $\set{d^h(u, s)}\cup\set{d(u, v')+d^h(v', s)}_{v'\in M}$ is either infinite, or corresponds to some (not necessary simple) path, thus, since the graph has non-negative weights, it holds that $\tilde{d}(u, s)\geq d(u, s)$. We show that $\tilde{d}(u, s)\leq d(u, s)$ by considering two cases. If shortest path between $u$ and $s$ has less than $h$ hops, then $d(u, s)=d^h(u, s)$ appears as an argument to the set over which we take the minimum and thus $\tilde{d}(u, s)\leq d^h(u, s)=d(u, s)$. Otherwise, by Property~\ref{itm:skeleton:sp} \whp there exists a shortest path between $u$ and $s$ such that there is a node $v''\in M$ in one of its $h$ last (closest to $s$) nodes.
        
        Thus, the corresponding element $d(u, v'')+d^h(v'', s)=d(u, s)$ appears in $\set{d(u, v')+d^h(v', s)}_{v'\in M}$ for $v'=v''$. Therefore, in both cases $\tilde{d}(u, s)\leq d(u, s)$. 
    \end{proofof}

\subsection{Approximating  \texorpdfstring{$n^x$-}{M}SSP from a Given Set of Sources}\label{app:MSSP}
We previously showed how to approximate distances to a set of independently, uniformly, randomly chosen sources, and here we leverage this to the case of a \emph{given} set of sources, which is the case we care more about. We show an algorithm for approximating distances from a given set of sources, which is \emph{tight}, matching the lower bound of \cite{kuhn2020computing}, when the number of sources is at least $\Omega(n^{2/3})$: given $O(n^x)$ sources, our algorithm takes $\tilde O(n^{1/3}/\epsilon + n^{x/2})$ rounds for a $(1 + \epsilon)$-approximation in unweighted graphs, and a $3$-approximation in weighted graphs.

The proof of the following, appears as a part of \cite[Theorem 4.2]{kuhn2020computing}, which shows how to obtain distance approximation algorithms in the \hybrid model from distance approximation algorithms in the \clique model. A similar approach is also used in the proof of \cite[Theorem 2.3]{AHKSS20}. We find this tool useful in general in order to obtain, given some algorithm which approximates distances on a skeleton graph, an algorithm for approximating distances from a set of nodes to the entire graph, and hence we extract it here from the proof of \cite[Theorem 4.2]{kuhn2020computing}.

\begin{claim}[$n^{x}$-SSP from APSP on Skeleton Graph]\cite[Theorem 4.2]{kuhn2020computing}\label{thr:MSSPToRAPSP}
    Consider a graph $G=(V, E)$, its skeleton graph $S_x=(M, E_x)$, for some constant $0<x<1$ and a set of sources $U$, where $|U| = \tilde \Theta(n^y)$, for some constant $0 < y < 1$. Let $A_x$ be a $T$-round algorithm in the \hybrid model, such that given a skeleton graph $S_x$, $A_x$ ensures that for every pair $v, v'\in M$, both $v$ and $v'$ know an $(\alpha, \beta)$-approximation for the distance between them in $G$. 

    Then, for any value $\eta\geq 1$, there is an algorithm $B$ for the \hybrid model which takes $\tildeBigO{T+n^{\frac{y}{2}}+\eta\cdot n^{1-x}}$ rounds over $G$, and ensures that for every $v \in V, s \in U$, node $v$ knows an approximation for the distance from $v$ to $s$ in $G$, where the approximation factor is $\left(2\alpha+1, \beta\right)$ if $G$ is weighted, and $\left(\alpha + \frac{2}{\eta}, \beta\right)$ if $G$ is unweighted. 
\end{claim}

\begin{claimproof}
    We follow the proof of \cite[Theorem 2.2]{kuhn2020computing}.
    We apply $A_x$ on $S$, and denote by $\tilde{d}^S$ the computed $(\alpha,\beta)$-approximation for the distances amongst the set of nodes $M$.
    Afterwards, nodes in $V$ learn their $(\eta\cdot h)$-hop neighborhoods, where $h=\tildeBigO{n^{1-x}}$. As in \cite[Theorem 4.2]{kuhn2020computing}, each source $s\in U$ selects its \emph{representative} -- the closest node $r_s\in M$ found in its $h$-hop neighborhood -- and participates in a token dissemination protocol with the token $\Braket{r_s, d^{h}(r_s, s)}$. Each node $u\in V$, for each source $s\in U$, outputs $\tilde{d}(u, s)=\min\set{d^{\eta h}(u, s), \min_{v'\in M\cap N_G^h(u)}\set{d^{h}(u, v')+\tilde{d}(v', r_s)} + d^{h}(r_s, s)}$.
    
    Each node $u\in V$ knows $\tilde{d}(v', r_s)$ since it is an output of the algorithm $A_x$ for nodes $v'\in M\cap N_G^h(u)$ and $u\in V$ learns this information from each $v'$ in its $h$-hop neighborhood. Node $u$ also knows $d^{h}(r_s, s)$ \whp, due to \fullref{thr:tokenDissemination}.
    
    There are at most \tildeBigO{n^{y}} source nodes \whp due to Property~\ref{itm:skeleton:size}, so the token dissemination protocol with $k=\tildeBigO{n^{y}}$ and $\ell=1$ takes $\tildeBigO{\sqrt{n^y}}=\tildeBigO{n^{\frac{y}{2}}}$ rounds \whp, by \fullref{thr:tokenDissemination}. Overall, the number of rounds is \tildeBigO{T+n^{\frac{y}{2}}+\eta\cdot n^{1-x}}, as needed.
    
    To show the approximation, first notice that the algorithm does not underestimate the distances, since each value in the set over which the minimum is taken is either infinite or corresponds to the approximate weight of a not necessarily simple path. This also implies that in case shortest path between $u$ and $s$ has less than $\eta h$ hops, it holds that $\tilde{d}(u, s)=d^{\eta h}(u, s)=d(u, v)$, so we assume without loss of generality that shortest path between $u$ and $v$ at at least $h$ hops and also $d(u, s)\geq hop(u, s)>\eta\cdot h$.

    Take two nodes $u\in V,s\in U$, by Property~\ref{itm:skeleton:sp}, there is a shortest path between $u$ and $s$ such that there is a node $u'\in M$ in the first $h$ nodes (from the $u$-side) and a node $v''\in M$ in the last ${h}$ nodes (from the $v$ side). Let $v'=\arg\min_{v'\in M}\set{d^{h}(u, v')+\tilde{d}(v', r_s)} + d^{h}(r_s, s)$.
    
    We upper bound $\tilde{d}$ for the weighted case:
    \begin{align*}
        \tilde{d}(u, s)\leq&
        \left(d^{h}(u, v')+\tilde{d}(v', r_s)\right)+d_{ h}(r_s, s)\\
        &\leq\left(d(u, u')+\tilde{d}(u', v'')\right)+d(v'',s)\\
        &\leq\left(d(u, u')+d(v'',s)\right)+ \alpha \cdot d(u', v'') + \beta\\
        &\leq d(u, s)+\alpha \cdot \left(d(u', s)+d(s, r_s)\right) + \beta\\
        &\leq d(u, s)+\alpha \cdot \left(d(u, s)+d(v, v'')\right) + \beta\\
        &\leq\left(2\alpha + 1\right)\cdot d(u, s) + \beta,
    \end{align*}
    where the first transition is due to the computation of $\tilde{d}(u, s)$, the second is implied by the choice of $v'$ and the representative $r_s$, the third follows from the definition of an $(\alpha,\beta)$-approximation of the distances and Property~\ref{itm:skeleton:distance} of the skeleton graph, the forth is due to the non-negative weights and the triangle inequality, the fifth is due to the non-negative weights and the choice of the representative $r_s$, and the last one is due to the non-negative weights. Also, we may use the fact that $d(u, v)\geq \eta \cdot h$ and obtain a purely multiplicative approximation factor of $2\alpha + 1 + \frac{\beta}{\eta \cdot h}$.
    
    However, for the unweighted case we can upper bound $\tilde{d}$ slightly better:
    \begin{align*}
        \tilde{d}(u, s)\leq&
        \left(du, v')+\tilde{d}^S(v', r_s)\right)+d(r_s, s)\\
        &\leq\left(d(u, u')+\tilde{d}^S(u', v')\right)+d(r_s, s)\\
        &=2\cdot h + \alpha \cdot d(u', v') + \beta\\
        &<2\cdot \frac{d(u, v)}{\eta} + \alpha \cdot d(u, v) + \beta,
    \end{align*}
    where the first transition is due to the computation of $\tilde{d}(u, s)$, the second is implied by the choice of $v'$, the third follows from $u'$ and $r_s$ being in $h$-hop neighborhood of $u$ and $s$, respectively, and the definition of a $(\alpha,\beta)$-approximation, and the forth is due to the assumption that $d(u, s)> \eta\cdot h$.
    Again, using the fact that $d(u, v)\geq \eta \cdot h$, we obtain a purely multiplicative approximation factor of $\alpha + \frac{2}{\eta} + \frac{\beta}{\eta\cdot h}$.
\end{claimproof}

We use the claim above to compute our tight approximation of the $n^{y}$-SSP.

\MSSP*
\begin{proof}
    We use $x = 2/3$ and let $A_{2/3}$ be the algorithm from \fullref{thr:exactRAPSP}. Notice that \fullref{thr:exactRAPSP} expects a skeleton graph and not a random set $M$, yet it is possible to convert $M$ to a skeleton graph using \fullref{claim:skeletonOnSampled} in $\tilde O(n^{1/3})$ rounds \whp Further, \fullref{thr:exactRAPSP} computes the distances between the nodes of $M$ over a skeleton graph and not over $G$, as required of $A_{2/3}$ in \fullref{thr:MSSPToRAPSP}, yet, this is equivalent due to Property~\ref{itm:skeleton:distance} in the definition of a skeleton graph. Using $A_{2/3}$, the rest of the proof follows directly from \fullref{thr:MSSPToRAPSP}, with $\alpha=1,\beta=0,\eta=2/\epsilon,x=2/3,T=\tilde O(n^{1/3})$ in $\tildeBigO{n^{1/3}+n^{y/2}+n^{1/3}/\epsilon}=\tildeBigO{n^{1/3}/\epsilon+n^{y/2}}$ rounds leading to the approximation factor of $\left(2\cdot 1 + 1=3, 0\right)=(3, 0)$ for weighted graphs and $\left(1+2/(2/\epsilon), 0\right)=(1+\epsilon, 0)$.
\end{proof}

Notice, that for $y\geq \frac{2}{3}$, the complexity is $\tildeTheta{n^{\frac{y}{2}}}$, which is tight due to the lower bound of \cite[Theorem 1.5]{kuhn2020computing}.

\subsection{Approximations of Eccentricities}\label{app:ecc}

\begin{restatable}[$(1+\epsilon)$-Approx. Unweighted Eccentricities]{lemma}{unweightedEccentricities}\label{thr:unweightedEccentricitiesApproximation}
Let $G=(V, E)$ be an unweighted graph, and let $\epsilon>0$. Then there exists an algorithm in the \hybrid model which, for each $v$, computes a $(1+\epsilon)$-approximation of unweighted $ecc(v)$ in $\tildeBigO{n^{1/3}/\epsilon}$ rounds \whp
\end{restatable}
\begin{proof}
    We run the algorithm from \fullref{thr:exactRMSSP} with $x=2/3$ to obtain $n^{2/3}$-RSSP from a set of nodes $M$ selected independently randomly with probability $n^{-1/3}$ from $V$, such that every $v \in V$ knows its distance to every node in $M$. This takes $\tilde O(n^{1/3})$ rounds. Then, each node $v$ learns its $\left(1+\frac{1}{\epsilon}\right)h$-hop neighborhood, where $h=\tildeTheta{n^{1/3}}$, in $\tilde O(n^{1/3}/\epsilon)$ rounds. The approximate eccentricity $\widetilde{ecc}(v)$ of node $v$ is the maximum between the distance to the farthest node in $M$ from $v$, and distance to the farthest node in the $\left(1+\frac{1}{\epsilon}\right)h$-hop neighborhood of $v$.
    
    We now prove the approximation factors. Notice, that there exists a node $w \in V$ such that $\widetilde{ecc}(v)=d(v, w)\leq ecc(v)$, so we do not overestimate the eccentricity. If $ecc(v)\leq \left(1+\frac{1}{\epsilon}\right) h$, the farthest node is in the explored neighborhood and the returned $\widetilde{ecc}(v)$ is exact in this case. Let $u$ be farthest node from $v$. By \cref{claim:skeletonOnSampled}~and~Property~\ref{itm:skeleton:sp}
    
    there exists a node $w\in M$ on some shortest path between $v$ and $u$ within hop-distance at most $h$ from $u$ \whp By definition, $\widetilde{ecc}(v)
    \geq d(v, w)
    \geq ecc(v)-h
    \geq ecc(v)-\frac{ecc(v)}{1+\frac{1}{\epsilon}}
    =\left(1-\frac{1}{\frac{\epsilon+1}{\epsilon}}\right)ecc(v)
    =\left(1-\frac{\epsilon}{1+\epsilon}\right)ecc(v)
    =\left(\frac{1+\epsilon - \epsilon}{1+\epsilon}\right)ecc(v)
    =\frac{1}{1+\epsilon}ecc(v)$.
\end{proof}

\begin{restatable}[$3$-Approx. Weighted Eccentricities]{lemma}{weightedEccentricities}\label{thr:weightedEccentricitiesApproximation}
Let $G=(V, E)$ be a weighted graph. There is an algorithm in the \hybrid model that computes a $3$-approximation of weighted eccentricities in \tildeBigO{n^{1/3}} rounds, \whp 
\end{restatable}
\begin{proof}
    We run the algorithm from \fullref{thr:exactRMSSP} with $x=2/3$ to obtain $n^{2/3}$-RSSP from a set of nodes $M$ selected independently randomly with probability $n^{-1/3}$ from $V$, such that every $u \in V$ knows its distance to every node in $M$. Then, each $v\in M$ learns its $h$-hop neighborhood, where $h=\tildeTheta{n^{1/3}}$, and we use token dissemination, \fullref{thr:tokenDissemination}, to ensure that all the nodes in the graph know $ecc_h(v)=\max_{w\in N_G^h(v)}\set{d(v, w)}$ for each $v \in M$. The approximate eccentricity $\widetilde{ecc}(u)$ of node $u\in V$, is the value of \cref{eq:eccentricities}.
    \begin{align}\label{eq:eccentricities}
        \widetilde{ecc}(u)=\frac{1}{3}\max_{v\in M}\set{d(u, v)+ecc_h(v)}
    \end{align}
     
    To show the approximation factor for some $u\in V$, let $v'= \arg\max_{v\in M}\set{d(u, v)+ecc_h(v)}$ and $v''\in N_G^h(v)$ be the farthest node from $v'$ in its $h$-hop neighborhood, that is $ecc_h(v')=d(v', v'')$. Further, let $p\in V$ be the farthest node from $u\in V$, in other words, $ecc(u) = d(u, p)$, and $p'\in M$ be some marked node in the $h$-hop neighborhood of $p$, which exists \whp by Chernoff bound. Finally, let $p''\in V$ be the farthest node from $p'$ in its $h$-hop neighborhood, that is, $ecc_h(p')=d(p', p'')$. Notice that 
    $\widetilde{ecc}(u)
    =\frac{d(u, v')+d(v', v'')}{3}
    \geq \frac{d(u, p')+d(p', p'')}{3}
    \geq \frac{d(u, p')+d(p', p)}{3}
    \geq \frac{d(u, p)}{3}
    =\frac{ecc(u)}{3}$, where the first inequality is due to the choice of $v'$, the second inequality follows from the choice of $p''$ and the third is implied by triangle inequality. Also, we do not overestimate the eccentricity, since $\widetilde{ecc}(u)
    =   \frac{d(u, v')+d(v', v'')}{3}
    \leq \frac{d(u, v') + d(v', u) + d(u, v'')}{3}
    \leq ecc(u)$, where the first inequality is due to triangle inequality and the second is due to the graph being undirected and due to the definition of eccentricity.
     
    By \fullref{thr:exactRMSSP} the round complexity of the $n^{2/3}$-RSSP is $\tildeBigO{n^{1/3}}$ \whp and token dissemination with a total of $k=\size{M}=\tildeBigO{n^{2/3}}$ tokens \whp (by Chernoff bound) and $\ell=1$ tokens per node takes $\tildeBigO{n^{1/3}}$ rounds \whp due to \fullref{thr:tokenDissemination}.
\end{proof}

Combining \fullref{thr:unweightedEccentricitiesApproximation,thr:weightedEccentricitiesApproximation} gives \fullref{thr:eccentricitiesApproximation}.

\subsection{Diameter Approximations}\label{app:diameter}
\unweigtedDiameterApproximation*
\begin{proof}
    To achieve this result we first use the algorithm from \fullref{thr:unweightedEccentricitiesApproximation} to compute $(1+\epsilon)$-approximate eccentricities, for each $v \in V$ a value $\widetilde{ecc}(v)$, in $\tildeBigO{n^{1/3}/\epsilon}$ rounds, and then use \fullref{thr:aggregateAndBroadcast} to compute the maximum between the approximate eccentricities $\widetilde{ecc}(v)$ in an additional number of $\tilde O(1)$ rounds.
    Notice that the maximum of $\left(1+\epsilon\right)$-approximate eccentricities is indeed a good approximation of the diameter, as 
    \begin{align*}
        \widetilde{D}=
        \max\Set{\widetilde{ecc}(v)}\geq 
        \max\Set{\frac{ecc(v)}{1+\epsilon}}=
        \frac{\max\Set{ecc(v)}}{1+\epsilon}=
        \frac{D}{1+\epsilon}.
    \end{align*}

    On the other hand, since the approximate eccentricities $\widetilde{ecc}(v)$ do not overestimate the real eccentricities, their maximum does not over overestimate the true diameter. Thus, $\widetilde{D}
        =\max\Set{\widetilde{ecc}(v)}
        \leq\max\Set{{ecc}(v)}=D$, which completes the proof.
\end{proof}

Our oracle-based techniques allow us to solve weighted SSSP fast. After we do it, the following well-known simple reduction allows us to compute $2$ approximate weighted diameter.
\begin{restatable}[Diameter From SSSP]{claim}{DfromSSSP}\label{thr:ssspToDiameter}
Given a graph $G = (V, E)$, a value $\alpha>0$ and an algorithm which computes an $\alpha$ approximation of weighted SSSP in $T$ rounds of the \hybrid model, there is an algorithm which computes a $2\alpha$-approximation of the weighted diameter in $T+\tildeBigO{1}$ rounds of the \hybrid model.
\end{restatable}
\begin{claimproof}
    We compute SSSP from an arbitrary node $s$. Then, each node $v$ proposes its candidate for the diameter $\frac{1}{\alpha}\tilde{d}(v, s)$, which it computes after SSSP. The approximation is the maximum of all the proposals $\tilde{D}=\max_{u\in V}{\frac{1}{\alpha}\tilde{d}(v, s)}$, and us computed using \fullref{thr:aggregateAndBroadcast}.
    
    Since for each $v\in V$, it holds that $\tilde{d}(v, s)\leq \alpha d(v, s)$ is a shortest path between $v$ and $s$, $D\geq \tilde{D}=\max_{u\in U}{\frac{1}{\alpha}d(v, s)}$. Let $u$ and $v$ be the endpoints of some path whose length is the diameter (i.e., $D=d(u, v)$), by the triangle inequality, $d(u, s)+d(s, v)\geq d(u, v)$. Thus, at least one of the terms is greater then $\frac{d(u, v)}{2}$, assume without loss of generality that $d(u, s)\geq \frac{d(u, v)}{2}$. Since $\tilde{d}(u, s)\geq d(u, s)$ and the proposal of $u$ is at least $\frac{1}{\alpha}\tilde{d}(u, s)\geq \frac{1}{2\alpha}d(u, s)$, it holds that $\tilde{D}\geq d(u, s)\geq \frac{d(u, v)}{2}=\frac{D}{2\alpha}$.
    
    The round complexity for computing the SSSP approximation is $T$, and the round complexity for the aggregation operation is \tildeBigO{1} by \fullref{thr:aggregateAndBroadcast}.
\end{claimproof}

\weightedDiameter*
\begin{proof}
    The claim follows immediately from \fullref{theorem:exactSSSP,thr:ssspToDiameter} with $\alpha=1$.
\end{proof}

\section{A Lower Bound for  \texorpdfstring{$n^x$}{n to the power of x}-RSSP}\label{sec:lowerBound}

We use the following \fullref{thr:lowerBoundToLearnRandomVariable} to obtain our lower bound. 
\begin{claim}[Lower Bound Random Variable]\cite[Lemma 4.12]{AHKSS20}\label{thr:lowerBoundToLearnRandomVariable}
Let $G = \left(V, E\right)$ be an $n$-node graph that consists of a subgraph $G' = \left(V', E'\right)$ and a path of length $L$ (edges) from some node $a \in V \setminus V'$ to $b \in V'$ and that except for node $b$ is vertex-disjoint from $V'$. Assume further that the nodes in $V$ are collectively given the state of some random variable $X$ and that
node $a$ needs to learn the state of $X$. Every randomized algorithm that solves this problem in the \hybrid network model requires $\bigOmega{\min\set{L,\frac{H(X)}{L\cdot\log^2n}}}$ rounds, where $H(X)$ denotes the Shannon entropy of $X$.
\end{claim}

\lowerBound*
\begin{proof}
    We use the same base construction as in~\cite[Theorem 2.5]{AHKSS20} and in~\cite[Theorem 1.5]{kuhn2020computing} with only slight modifications. However, for the sake of self-containment, we recall here the entire construction. 
    
    We construct an unweighted graph $G=\left(V, E\right)$ which contains a path, with endpoints at nodes $a$ and $c$, and two \emph{fooling} sets of nodes $S_b,S_c$, both of size $y=\floor*{\frac{n}{4}}$. Every node in $S_c$ has an edge to $c$. At hop-distance $L=\floor*{\frac{\sqrt{p\cdot n}}{\log n}}\geq 1$ from $a$ there is a node $b$. Each node in $S_b$ has an edge to $b$. The segment of path between $b$ and $c$ contains all other nodes, thus the hop distance between $b$ and $c$ is $z=n - 2y - L - 1$, and the hop distance between $a$ and $c$ is $x=z+L=n - 2y - 1$. See Figure~\ref{fig:lb}.

    \begin{figure}
        \begin{tikzpicture}
            \node[main node] (a) {$a$};
            \node (dotsab) [right=of a] {$\cdots$};
            \node[main node] (b) [right=of dotsab] {$b$};
            \node (dotsbc) [right=3cm of b] {$\cdots$};
            \node[main node] (c) [right=3cm of dotsbc] {$c$};
            \node[main node] (b1) [below left=of b] {$u_{b_1}$};
            \node (dotsb) [below=of b] {$\cdots$};
            \node[main node] (by) [below right=of b] {$u_{b_y}$};
            \node[main node] (c1) [below left=of c] {$u_{c_1}$};
            \node (dotsc) [below=of c] {$\cdots$};
            \node[main node] (cy) [below right=of c] {$u_{c_y}$};
            
            \path[draw]   
                (a) edge node {} (dotsab)
                (dotsab) edge node {} (b)
                (b) edge node {} (dotsbc)
                (dotsbc) edge node {} (c)
                (b1) edge node {} (b)
                (by) edge node {} (b)
                (c1) edge node {} (c)
                (cy) edge node {} (c)
            ;
        \end{tikzpicture}
        \caption{Lower bound construction.}
        \label{fig:lb}
    \end{figure}
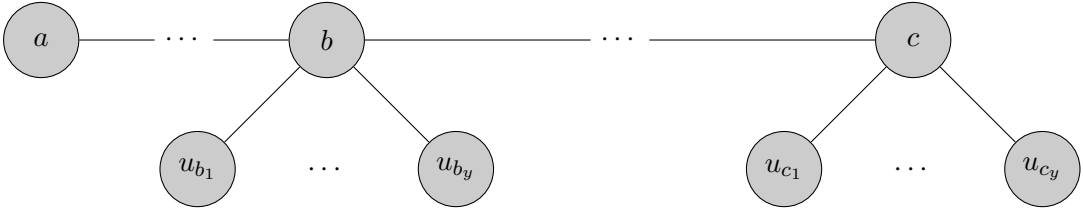
    
    We follow the scheme of \cite[Theorem 2.5]{AHKSS20} and reserve $2y$ identifiers from the set $\interval{2y}$ for nodes in $S=S_b\cup S_c$ and assign additional $n-2y$ identifiers to other nodes $V\setminus S$ in a globally known manner. The nodes from $S=S_b\cup S_c$ are assigned identifiers randomly from the reserved set $\interval{2y}$, as follows. Each identifier from $\interval{y}$ (the first half of identifiers) is assigned with probability $1/2$ to the next currently unlabeled node in $S_b$ and with probability $1/2$ to the next currently unlabeled node in $S_c$. The other $y$ reserved identifiers are arbitrarily assigned to nodes left unlabeled in $S$.
   
    Let $S'$ be the set of nodes which are sampled to be sources. Consider the set $S''=\interval{y}\cap S'$ of sources with identifiers in $\interval{y}$. For $1\leq i\leq y$ we define $X_i$ to be an indicator of whether identifier $i$ is one of the sampled source identifiers. Notice that in expectation there are $E\left[\sum_{i=1}^{y}X_i\right]=p\cdot y=p\floor*{\frac{n}{4}}$ sampled nodes with identifiers in $\interval{y}$. Moreover, for a large enough $n$, since $p=\bigOmega{\frac{\log n}{n}}$, by Chernoff Bounds, there exists a constant $\beta > 0$ such that \whp  $\sum_{i=1}^{y}X_i\geq \beta p\cdot n$, which means that there are at least $\beta\cdot p \cdot n=\bigOmega{p\cdot n}$ source nodes in $S''$. These may be in $S_b$ or $S_c$. 
    
    By definition, an $\alpha$-approximation of distances from sources $S'$ is a function $\widetilde{hop}\colon \left(V, S'\right) \mapsto \Z$ which satisfies for every $v\in V,u\in S'\colon hop(v, u)\leq \widetilde{hop}(v, u)\leq \alpha hop(v, u)$, where $hop(u,v)$ denotes the hop distance between $u$ and $v$. In particular, the node $a\in V$ should compute the values $\widetilde{hop}\left(a, \cdot\right)$, which satisfy for every $u\in S'':hop(a, u)\leq \widetilde{hop}(a, u)\leq \alpha hop(a, u)$. If $u\in S_c$ then $hop(a, u)=x+1$, otherwise, if $u\in S_b$ then $hop(a, u)=L+1$. Thus, if $a$ does not know whether $u\in S_c$ or $u\in S_b$, it has to be on the safe side and return $\widetilde{\delta}\left(a,u\right)=x+1\geq L+1$. The approximation factor for $u\in S_b$ is     
    \[
    \alpha
    =\frac{x + 1}{L+1}
    =\frac{n - 2y}{L+1}
    =\frac{n - 2\floor*{\frac{n}{4}}}{\frac{\sqrt{pn}}{\log n}}
    \geq\frac{1}{2}\sqrt{\frac{n}{p}}\log n.
    \]
    
    Given $S'$, let $Y\in\Set{b, c}^{\size{S''}}$ be a random variable indicating for each sampled node $S''$ whether it belongs to $S_b$ or $S_c$. As in the proof of \cite[Theorem 2.5]{AHKSS20}, to get any better approximation, node $a$ has to learn for each node $s\in S''$ whether it is in $S_b$ or $S_c$.  For this $a$ has to learn the random variable $Y$ whose entropy is \whp \[
        H\left(Y\right)
        =\sum_{q\in \set{b, c}^{\size{S''}}}{{2^{-q}}\log\left(2^q\right)}
        =\bigOmega{\size{S''}}
        =\bigOmega{p\cdot n}.
    \] 
    By \fullref{thr:lowerBoundToLearnRandomVariable}, \bigOmega{\frac{\sqrt{p\cdot n}}{\log n}} rounds are required for $a$ to learn the variable $Y$ \whp 
\end{proof}

\end{document}